\documentclass[a4paper,11pt,centertags,reqno]{amsart}

\usepackage[foot]{amsaddr}
\usepackage{letltxmacro}        						
\usepackage{ifthen}													
\usepackage[headings]{fullpage}					
\usepackage[nodisplayskipstretch]{setspace}						
\usepackage[inline]{enumitem}   							
\usepackage{array}														
\usepackage{fancybox}													    
\usepackage{graphicx}           					
\usepackage{amsmath}	        	
\usepackage{amsthm}							
\usepackage{amssymb}			
\usepackage{mathtools}					
\usepackage{mathrsfs}												
\usepackage{yhmath}							
\usepackage{accents}						       
\usepackage{xfrac}             
\usepackage{url}
\usepackage[all]{xy}						
\usepackage[normalem]{ulem}
\usepackage[round,comma,authoryear,sort&compress]{natbib} 
\usepackage{lmodern}
\usepackage{anyfontsize}
\PassOptionsToPackage{hyphens}{url}
\usepackage[colorlinks,citecolor=blue,urlcolor=blue,unicode,verbose]{hyperref}
\DeclareFontFamily{U}{mathx}{\hyphenchar\font45}
\DeclareFontShape{U}{mathx}{m}{n}{
      <5> <6> <7> <8> <9> <10>
      <10.95> <12> <14.4> <17.28> <20.74> <24.88>
      mathx10
      }{}
\DeclareSymbolFont{mathx}{U}{mathx}{m}{n}
\DeclareFontSubstitution{U}{mathx}{m}{n}
\DeclareMathAccent{\widecheck}{0}{mathx}{"71}
\DeclareMathAccent{\wideparen}{0}{mathx}{"75}

\makeatletter
\newcommand*\rel@kern[1]{\kern#1\dimexpr\macc@kerna}
\newcommand*\widebar[1]{%
  \begingroup
  \def\mathaccent##1##2{%
    \rel@kern{0.8}%
    \overline{\rel@kern{-0.8}\macc@nucleus\rel@kern{0.2}}%
    \rel@kern{-0.2}%
  }%
  \macc@depth\@ne
  \let\math@bgroup\@empty \let\math@egroup\macc@set@skewchar
  \mathsurround\z@ \frozen@everymath{\mathgroup\macc@group\relax}%
  \macc@set@skewchar\relax
  \let\mathaccentV\macc@nested@a
  \macc@nested@a\relax111{#1}%
  \endgroup
}
\makeatother

\newcommand{\N}{\mathbb{N}}

\newcommand{\R}{\mathbb{R}}

\newcommand{\Cx}{\mathbb{C}}

\newcommand{\Rinf}{\widebar{\R}}
\newcommand{\Cinf}{\widebar{\Cx}}
\renewcommand{\Re}{\operatorname{Re}}

\newcommand{\id}{{\operatorname{id}}}

\newcommand{\I}{\mathfrak{I}}
\renewcommand{\U}{\mathcal{U}}
\newcommand{\NaN}{\mathrm{NaN}}

\renewcommand{\log}{\ln}

\renewcommand{\C}[2]{
  \ifthenelse{#1=0 \and #2=0}{\textsf{\upshape C}}
  {\ifthenelse{#2=0}{\textsf{\upshape C}^{#1}}
    {\textsf{\upshape C}^{#1,#2}}}
}

\renewcommand{\d}{\mathrm{d}}

\newcommand{\e}{\mathrm{e}}

\newcommand{\E}{\textsf{\upshape E}}

\renewcommand{\P}{\textsf{\upshape P}}
\newcommand{\Qu}{\textsf{\upshape Q}}

\newcommand{\indicator}[1]{\mathbf{1}_{#1}}

\newcommand{\filt}[1]{\mathfrak{#1}}
\newcommand{\sigalg}[1]{\mathscr{#1}}

\newcommand{\Exp}{\ensuremath{\mathscr{E}}}
\newcommand{\Log}{\ensuremath{\mathcal{L}}}

\newcommand{\lc}{[\![}

\newcommand{\ccol}[1]{\omit\hfil$\displaystyle#1$\hfil\ignorespaces}
\newcommand*{\bigs}[1]{\scalebox{1.2}{\ensuremath#1}}

\DeclarePairedDelimiterX\br[1]{[}{]}{#1}

\DeclarePairedDelimiterX\set[2]\{\}{#1\::\:#2}

\makeatletter
\let\oldr@@t\r@@t
\def\r@@t#1#2{%
  \setbox0=\hbox{$\oldr@@t#1{#2\,}$}\dimen0=\ht0
  \advance\dimen0-0.2\ht0
  \setbox2=\hbox{\vrule height\ht0 depth -\dimen0}%
  {\box0\lower0.4pt\box2}}
\LetLtxMacro{\oldsqrt}{\sqrt}
\renewcommand*{\sqrt}[2][\ ]{\oldsqrt[#1]{#2}}
\makeatother

\theoremstyle{plain}
\newtheorem{theorem}{Theorem}
\newtheorem{lemma}[theorem]{Lemma}
\newtheorem{proposition}[theorem]{Proposition}
\newtheorem{corollary}[theorem]{Corollary}
\theoremstyle{definition}
\newtheorem{definition}[theorem]{Definition}

\newtheorem{example}[theorem]{Example}

\theoremstyle{remark}
\newtheorem{remark}[theorem]{Remark}

\numberwithin{theorem}{section}
\numberwithin{equation}{section}
\numberwithin{figure}{section}
\numberwithin{table}{section}

\graphicspath{{Figures/}}
\bibpunct{(}{)}{,}{a}{}{,}

\allowdisplaybreaks[4]
\makeatletter
\@namedef{subjclassname@2010}{MSC {[2010]}}
\makeatother
\newcommand\shorturl[1]{%
  \href{http://#1}{\nolinkurl{#1}}%
}

\begin{document}
\title[Simplified Stochastic Calculus With Applications]
{Simplified Stochastic Calculus With Applications in Economics and Finance}

\author{Ale\v{s} \v{C}ern\'{y}}

\address{Business School (formerly Cass)\\City, University of London\\
106 Bunhill Row, London, EC1Y 8TZ, UK.}

\email{ales.cerny.1@city.ac.uk}

\author{Johannes Ruf}
\address{Department of Mathematics\\
  LSE\\
	Houghton Street, London, WC2A 2AE, UK.}

\email{j.ruf@lse.ac.uk}

\thanks{Previous versions of this paper were circulated under the title ``Finance without Brownian motions: An introduction to simplified stochastic calculus.'' We thank Jan Kallsen, Jan-Frederik Mai, Lola Martinez Miranda, Johannes Muhle-Karbe, and two anonymous referees for helpful comments.}   


\begin{abstract} 
The paper introduces a simple way of recording and manipulating general stochastic processes without explicit reference to a probability measure. In the new calculus, operations traditionally presented in a measure-specific way are instead captured by tracing the behaviour of jumps (also when no jumps are physically present). The calculus is fail-safe in that, under minimal assumptions, all informal calculations yield mathematically well-defined stochastic processes. The calculus is also intuitive as it allows the user to pretend all jumps are of compound Poisson type. The new calculus is very effective when it comes to computing drifts and expected values that possibly involve a change of measure. Such drift calculations yield, for example, partial integro--differential equations, Hamilton--Jacobi--Bellman equations, Feynman--Kac formulae, or exponential moments needed in numerous applications. We provide several illustrations of the new technique, among them a novel result on the Margrabe option to exchange one defaultable asset for another.\medskip

\end{abstract}

\maketitle

\section{Introduction}\label{sect: intro}

Anyone who has attempted stochastic modelling with jumps will be aware of the sudden increase in mathematical complexity in models that are not of compound Poisson type. The difficulty is such that experienced researchers readily forgo generality in order to reduce the technical burden placed on their readers; see, for example, \citet{feng.linetsky.08.opre}, \citet{cai.kou.12}, \citet{hong.jin.18}, and \citet{ait-sahalia.matthys.19}.

In this paper we introduce an intuitive calculus that works for general processes but retains the simplicity of compound Poisson calculations. To achieve this, a change of paradigm is required. Classical It\^o calculus is based on decomposing the increments of every process into signal (drift, expected change) and noise (Brownian motion, zero-mean shock). This is at once convenient and mathematically expedient. The convenience of knowing the drift is immediate. Many tasks where stochastic processes are concerned involve computation of the drift of some quantity. Hamilton--Jacobi--Bellman equations in optimal control, for example, express the fact that the optimal value function plus the integrated historical cost is a martingale and therefore has zero drift. Similarly, Feynman--Kac formulae reflect zero drift of an integral of costs discounted at a specified stochastic killing rate. Closer to home, the Black--Scholes\nocite{black.scholes.73} partial differential equation can be obtained by setting the drift of the discounted option price process to zero under the risk-neutral measure. 

The expediency of the signal--noise decomposition comes from the early construction of the It\^o integral where the drift is integrated path-by-path but the Brownian motion integral is performed, loosely speaking, by summing up uncorrelated square-integrable random variables with zero mean. The paradigm shift is applied here: we separate how a process is recorded from the drift calculation. In other words, we do not carry the drift with us at all times but only evaluate it when the drift is really needed. This feels a little uncomfortable at first but there are ample rewards for the small intellectual effort required.

The consequences of the subtle change in perspective are far-reaching. By recording processes in a measure-invariant manner the technicalities of stochastic integration fall away, the importance of Brownian motions and Poisson processes recedes, and one begins to see
deeper into the fundamental relationships among the modelled variables, which now take center stage. Measure change, too, becomes an easy application of the simplified calculus, for as long as the new measure is directly driven by the variables being studied, which is overwhelmingly the case in practice.

We have prepared the paper with two audiences in mind. First and foremost, the paper is intended for the research community whose members do not consider themselves experts in mathematics in general, or stochastic analysis in particular, but who nevertheless use stochastic calculus as a modelling tool. To this readership we want to demonstrate that the new calculus is easy to understand and apply in practice.  Second, but no less important, we address colleagues specializing in stochastic analysis whom we wish to convince that all our arguments are mathematically rigorous.

The stated goal is not without its challenges. Where practical, we contrast the new approach with the more involved classical notation. In order to perform such comparison, one has to introduce some advanced concepts, such as the Poisson random measure, which are needed in classical stochastic calculus. Plainly, the lay readers will not be acquainted with some of the advanced concepts, nor do they need to be. Familiarity with Brownian motion, compound Poisson processes, some version of It\^o’s formula, and perhaps Girsanov’s theorem ought to be enough to sufficiently appreciate the backdrop against which the new calculus is constructed. The new calculus itself only needs a grasp of drift, volatility, and jump arrival intensity, plus three basic rules that are self-evident on an informal level.

The paper is organized as follows. In the rest of this introduction, we trace how the novel concept of this paper, the semimartingale representation \eqref{eq:190924.5}, arises from classical It\^o calculus. Section~\ref{sect: 2} provides a thorough introduction to the simplified stochastic calculus. It also explains how the proposed approach facilitates computation of drifts and expected values; in particular, it tackles the introductory example in the presence of jumps. Section~\ref{sect:190714.1} demonstrates the strength of the proposed approach on three additional examples. Section~\ref{sect:190620} amplifies this point by showcasing calculations that also require a change of measure. In particular, Example~\ref{E:181124.1} contains a new result that makes use of a non-equivalent change of measure. Section~\ref{sect: predictable} highlights the robustness of the proposed approach whereby, for a given task, the same representation applies in both discrete and continuous models. Such unification is unattainable in standard calculus. 

The examples in the paper are inspired by applications in Economics and Financial Mathematics but the broader lessons are clearly applicable to Science at large. We explore wider repercussions of the proposed methodology and briefly mention applications to Statistics and Engineering in the concluding Section~\ref{sect: conclusions}.
\subsection{McKean calculus for It\^o processes}\label{sect: Brownian}
For reasons of tractability, there is a preponderance of continuous-time stochastic models based on Brownian motion. Traditional stochastic calculus reflects this historical bias. As an example, consider a stochastic model for two economic variables, capital $K$ and labour $L$,
\begin{align}
\frac{\d K_t}{K_t} &= \mu_K \d t + \sigma_K \d W_t, \label{eq: K}\\
\frac{\d L_t}{L_t} &= \mu_L \d t + \sigma_L  \left(\rho_{KL}\d W_t+\sqrt{1-\rho^2_{KL}}\d \widehat{W}_t\right). \label{eq: L}
\end{align}
Here $W$ and $\widehat{W}$ are two independent Brownian motions. The inputs in this model are $\mu_K$, $\mu_L$, $\sigma_K>0$, $\sigma_L>0$, and $\rho_{KL}\in [-1,1]$ describing the correlation between the changes in $K$ and $L$. Informally, the `drift part' $\mu_K\d t$ represents expected change, while the `noise' $\sigma_K \d W_t$ is loosely interpreted as a shock with mean zero and variance $\sigma_K^2\d t$. 

The symbol $\d K_t$ represents an increase in capital over an infinitesimal time period $\d t$. The left-hand side of \eqref{eq: K} signifies percentage change in capital over the same period. The percentage change \emph{per unit of time} is not well-defined because the derivative $\sfrac{\d W_t}{\d t}$ does not exist. However, the \emph{expected} percentage change per unit of time is finite and equal to $\mu_K$.  This means the expected proportional increase in capital over a fixed time horizon $T$ equals 
\begin{equation}\label{eq: E[K_T]}
\E\left[\frac{K_T}{K_0}\right] = \e^{\mu_K T}.
\end{equation}
In the terminology of \citet{samuelson.65}, $K$ and $L$ are \emph{geometric Brownian motions} with drift rates $\mu_K$ and $\mu_L$, respectively.

Suppose we are interested in the evolution of the capital-labour ratio,  $\sfrac{K}{L}$. The standard It\^o calculus (\citealp[Theorem 6]{ito.51}; \citealp[Theorem~3.6]{karatzas.shreve.91}; \citealp[p.~95]{duffie.01}; \citealp[Theorem~4.16]{bjork.09}) yields, after simplifications, 
\begin{equation}\label{eq: ItoCondensed}
\begin{split}
\d \left(\frac{K_t}{L_t}\right) = & \frac{K_t}{L_t}(\mu_K -\mu_L -\rho_{KL}\sigma_K \sigma_L+ \sigma^2_L) \d t\\														
								&\quad+\frac{K_t}{L_t}(\sigma_K - \rho_{KL}\sigma_L )\d W_t - \frac{K_t}{L_t}\sigma_L \sqrt{1-\rho^2_{KL}}\d \widehat{W}_t.
\end{split}
\end{equation}

One can make two observations at this point:
\begin{enumerate}[label={\rm(\arabic{*})}, ref={\rm(\arabic{*})}] 
	\item\label{obs1} The formula \eqref{eq: ItoCondensed} is not easy to decipher --- the calculations are involved. 
	\item\label{obs2} The formula is `misleading' --- the processes $K$ and $L$ are already given, therefore the object on the left-hand side is defined path by path as the ratio $\sfrac{K_T(\omega)}{L_T(\omega)}$ and cannot depend on the reference probability measure. In contrast, some of the objects on the right-hand side are strongly measure-dependent: certainly, if we change the reference probability measure, there is no guarantee that $W$ and $\widehat{W}$ will still be Brownian motions under the new measure.
\end{enumerate}

\citet{mckean.69} addresses \ref{obs1} and informally also \ref{obs2} by rewriting \eqref{eq: ItoCondensed} in the form
\begin{align}
\d \left(\frac{K_t}{L_t}\right) = \frac{K_t}{L_t}\left(\frac{\d K_t}{K_t} - \frac{\d L_t}{L_t} -\frac{\d K_t}{K_t} \frac{\d L_t}{L_t}	+\left(\frac{\d L_t}{L_t}\right)^2\right) , \label{eq: shortIto}						 
\end{align}
where $\d K_t\d L_t$ is understood to stand for $\d [K,L]_t$ and $[K,L]$ is the quadratic covariation of the processes $K$ and $L$. In the present case we have 
$$\frac{\d [K,L]_t}{K_tL_t}=\rho_{KL}\sigma_K\sigma_L\d t \qquad \text{and}\qquad \frac{\d[L,L]_t}{L^2_t}= \sigma^2_L \d t.$$ 
Let us make two further observations:
\begin{enumerate}[resume, label={\rm(\arabic{*})}, ref={\rm(\arabic{*})}] 
	\item\label{obs3} Formula \eqref{eq: shortIto} is not only measure-independent, it is also model-free in the sense that it holds for \emph{any} two 
	continuous semimartingales $K$ and $L$ such that the integrals of $\sfrac{\d K_t}{K_t}$ and $\sfrac{\d L_t}{L_t}$ are well-defined.
	\item\label{obs4} \citet[p.~33]{mckean.69} observes that one can obtain \eqref{eq: shortIto} much more directly without passing through 
	\eqref{eq: K}--\eqref{eq: L}  
	and \eqref{eq: ItoCondensed}, simply by writing down a second-order Taylor expansion for $f(K,L)=\sfrac{K}{L}$ in the form
	\begin{equation}\label{eq:201020}
	\begin{split} 
		\d f(K_t,L_t) &= \tfrac{\partial f}{\partial k}(K_t,L_t)\d K_t +\tfrac{\partial f}{\partial \ell}(K_t,L_t)\d L_t \\
		&\ + \frac{1}{2}\left(\tfrac{\partial^2 f}{\partial k^2}(K_t,L_t)(\d K_t)^2+2\tfrac{\partial^2 f}{\partial k\partial \ell}(K_t,L_t)\d K_t\d L_t+\tfrac{\partial^2 f}{\partial \ell^2}(K_t,L_t)(\d L_t)^2\right).
	\end{split}
	\end{equation}
\end{enumerate}

Suppose now we wish to evaluate the \emph{expected} value of the capital--labour ratio. Here the formula \eqref{eq: ItoCondensed} is very helpful because it tells us immediately that $\sfrac{K}{L}$ is a geometric Brownian motion with the drift rate
\begin{equation}
b=\mu_K-\mu_L-\rho_{KL}\sigma_K\sigma_L+\sigma^2_L. \label{eq: drift Log(K/L)}
\end{equation}
With this coefficient in hand one swiftly concludes, in analogy to \eqref{eq: E[K_T]}, 
\begin{equation}\label{eq: E[K_T/L_T]}
\E\left[\frac{K_T}{L_T}\right]=\frac{K_0}{L_0} \e^{bT}.
\end{equation} 

The need for equation \eqref{eq: ItoCondensed} is only illusory, however. One can obtain formula \eqref{eq: drift Log(K/L)} equally easily from the measure-independent formula \eqref{eq: shortIto} by inserting the expected rate of change of $K, L$ and their quadratic (co)variations on the right-hand side of \eqref{eq: shortIto} as implied by \eqref{eq: K}--\eqref{eq: L}, 
\begin{alignat}{10}
 &\ccol{\dfrac{\d (\sfrac{K_t}{L_t})}{\sfrac{K_{t}}{L_{t}}}}  & \,\, = \,\, &  \ccol{\dfrac{\d K_t}{K_{t}}}  & \,\, -\,\,  &  \ccol{\dfrac{\d L_t}{L_{t}}}  & \,\, -\,\,  &  \ccol{\dfrac{\d[K,L]_t}{K_{t}L_t}}  & \,\, +\,\, & \ccol{\dfrac{\d [L,L]_t}{L_{t}^2}}, \label{eq: d(K/L) rigorous1}\\
 & \ccol{\downarrow}  &     &  \ccol{\downarrow}  &   &  \ccol{\downarrow}  & &  \ccol{\downarrow}  & & \ccol{\downarrow} &\nonumber\\[-0.1cm]
 & \ccol{b}   &  \,\, = \,\, &  \ccol{\mu_K}                        & \,\, -\,\, &  \ccol{\mu_L} & \,\, -\,\, & \ccol{\rho_{KL}\sigma_K\sigma_L}  & \,\, +\,\, & \ccol{\sigma^2_L.}  	\label{eq: drift}		 
\end{alignat}
Note that apart from the initial values $K_0$, $L_0$ and the time horizon $T$, the calculation requires \emph{five} characteristics of the capital and labour processes:  $\mu_K$, $\mu_L$, $\sigma_K$, $\sigma_L$, and $\rho_{KL}$.

\setstretch{1.2}
\subsection{First steps}\label{sect: first steps}
Let us now modify 
\eqref{eq: K}--\eqref{eq: L} 
by adding two jump components $J^K$ and $J^L$ that jointly form a L\'evy process, 
\begin{align}
\frac{\d K_t}{K_{t-}} &= \mu_K \d t + \sigma_K \d W_t + \d J^K_t; \label{eq: K jump}\smallskip\\
\frac{\d L_t}{L_{t-}} &= \mu_L \d t + \sigma_L  \left(\rho_{KL}\d W_t+\sqrt{1-\rho^2_{KL}}\d \widehat{W}_t\right) + \d J^L_t; \label{eq: L jump}\\
\d J^K_t &= \int_{|x_1|\leq 1}x_1(N(\d t,\d (x_1, x_2) )-\Pi(\d (x_1, x_2))\d t) + \int_{|x_1|> 1}x_1N(\d t,\d (x_1, x_2));
\label{eq:191023.1}\\
\d J^L_t &= \int_{|x_2|\leq 1}x_2(N(\d t,\d (x_1, x_2))-\Pi(\d (x_1, x_2))\d t) + \int_{|x_2|> 1}x_2N(\d t,\d (x_1, x_2)).
\label{eq:191023.2}
\end{align}
Here $N$ is a Poisson random measure and $\Pi$ the corresponding L\'evy measure. In particular, when $(J^K,J^L)$ form a compound Poisson process, i.e., when $\Pi(\R^2)<\infty$, the quantity $\Pi(\R^2)$ is the arrival intensity of jumps and 
$$x = (x_1, x_2) \mapsto \frac{\Pi((-\infty,x_1)\times (-\infty,x_2))}{\Pi(\R^2)} $$ 
is the bivariate cumulative distribution of jump sizes. The complicated expression in \eqref{eq:191023.2} is required to accommodate the models where the small jumps of $L$ have infinite variation while, at the same time, the large jumps of $L$ have infinite mean.

In the presence of jumps, by convention, the paths of $K$ and $L$ are assumed to be right-continuous with left limits.
That means the value of capital \emph{before} a jump is given by the left limit $K_{t-}$ while $K_t$ is the value \emph{after} a jump. The jump is naturally defined as the difference between the two, $\Delta K_t = K_t- K_{t-}$ and likewise for $L$; see Figure~\ref{fig: RCLL}.
\begin{figure}
	\centering
		\includegraphics[width=0.5\textwidth]{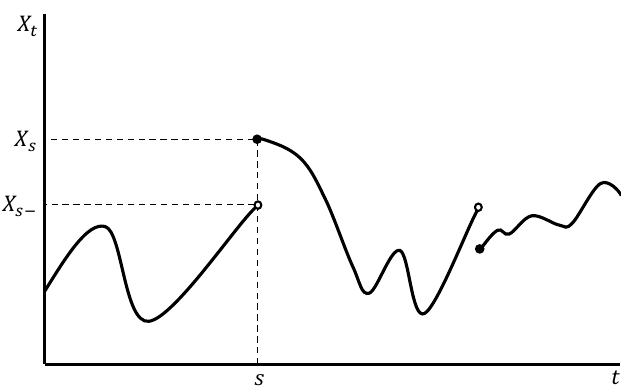}
	\caption{Illustration of a right-continuous path with left limits (process $X$).}
	\label{fig: RCLL}
\end{figure} 

Very little has changed between \eqref{eq: K}--\eqref{eq: L} and \eqref{eq: K jump}--\eqref{eq: L jump}; we have replaced one process with time-homogeneous independent increments by another. But unlike in the Brownian case, it is now mathematically possible that the right-hand side of \eqref{eq: L jump} --- the percentage change in labour supply $L$ --- does not have a finite first moment, while the mean of $\sfrac{K}{L}$ is nevertheless finite.\footnote{\label{fn:1}From a modelling point of view it is not realistic to believe that $L$ has infinite mean. We are simply stating that the calculus must be general enough to entertain such possibility. There are other circumstances where infinite mean arises more naturally, for example stock price $S$ will typically have finite mean but $\log S$ may plausibly have mean equal to $-\infty$ in a model with jumps if returns close to -100\% occur frequently enough.}

The main takeaway message is that decomposing a stochastic integral into `signal' and `noise' as suggested by 
\eqref{eq: K}--\eqref{eq: L} is, in general, not straightforward. One possibility is to split $J^L$ into two components containing the small and large jumps, respectively, and decompose only the small jump component into signal and noise as shown in \eqref{eq:191023.2}. This makes \eqref{eq: L jump} look more like \eqref{eq: L} and largely represents the current practice in applications where jumps of the most general type are considered, see \citet{kallsen.00}, \citet{fujiwara.miyahara.03},  \citet{hubalek.al.06}, \citet{jeanblanc.al.07,oksendahl.sulem.07}, \citet{bender.niethammer.08}, and \citet{applebaum.09}.

In this paper, we do handle arbitrary jumps but we interpret the difficulty with signal--noise decomposition very differently, which is to say we refrain from using such a decomposition altogether and instead look for a  measure-invariant representation of $\sfrac{K}{L}$. The sought expression must reduce to the McKean calculus \eqref{eq:201020} when $K$ and $L$ are continuous. At the same time, it must correctly account for changes due to jumps in $K$ and $L$. Just such a formula, suitably reinterpreted, can be traced to \citet[Section~3]{emery.78}. Let us now describe what we mean, first informally and then rigorously.

In the present example we seek the representation of $\sfrac{K}{L}$. This will be written symbolically as
\begin{equation}\label{eq:190924.0}
\d \left(\frac{K_t}{L_t}\right) = \frac{K_{t-}+\d K_t}{L_{t-}+\d L_t} - \frac{K_{t-}}{L_{t-}},
\end{equation}
which informally leads to an expression for percentage changes
\begin{equation}\label{eq:190924.1}
\frac{\d \left(\sfrac{K_t}{L_t}\right)}{\sfrac{K_{t-}}{L_{t-}}} 
=\frac{1+\sfrac{\d K_t}{K_{t-}}}{1+\sfrac{\d L_t}{L_{t-}}}-1.
\end{equation}
For \'Emery,  the right-hand side of \eqref{eq:190924.1} represents a deterministic, time-constant function
\begin{equation}\label{eq:190924.3}
\eta(x_1, x_2)=\frac{1+x_1}{1+x_2} - 1
\end{equation}
that acts on the increments $\sfrac{\d K_t}{K_{t-}}$ and $\sfrac{\d L_t}{L_{t-}}$.
The real meaning of the expression $\xi(\d X_t)$ for a  generic deterministic time-constant $\mathcal{C}^2$ function $\xi$ with $\xi(0)=0$ and any semimartingale $X$ is supplied by the \'Emery formula%
\footnote{The symbols $D\xi$ and $D^2\xi$ stand for first and second partial derivatives of $\xi$. \citet{emery.78} considers only the case when $\xi$ is deterministic and constant in time. We explicitly allow $\xi$ to be predictable in order to develop a calculus that includes stochastic integration and It\^o's formula.} 
\begin{equation}\label{eq:190924.5}
\xi_t(\d X_t) = D \xi_t (0) \d X_t+ \frac{1}{2} \sum_{i,j = 1}^d D^2_{ij} \xi_t(0)\, \d\bigs[X^{(i)},X^{(j)}\bigs]^c_t 
+ \left(\xi_t(\Delta X_t) -  D \xi_t(0) \Delta X_t  \right),
\end{equation}
where the last term yields absolutely summable jumps of finite variation and $[\cdot,\cdot]^c$ stands for the continuous part of quadratic covariation. An easy way to memorize the \'Emery formula is suggested in equation \eqref{eq:191202.3} and the subsequent paragraph. 

Having assigned meaning to the right-hand side of \eqref{eq:190924.1} via the formula \eqref{eq:190924.5} with $\xi=\eta$ from \eqref{eq:190924.3}, equality \eqref{eq:190924.1} is no longer an informal expression but a theorem whose validity one needs to establish. To accomplish this goal and build the simplified calculus, one can begin with the observation that \eqref{eq:190924.0}, too, represents a function that acts on the increments of the underlying processes; in this case it acts on $\d K_t$ and $\d L_t$. One important difference is that the function in question is no longer deterministic, i.e., we have  
$$\d \left(\frac{K_t}{L_t}\right) = \tilde\eta_t(\d K_t,\d L_t) $$
with 
\begin{equation}\label{eq:190925.1}
\tilde\eta_t(x_1, x_2)=\frac{K_{t-}+x_1}{L_{t-}+x_2} - \frac{K_{t-}}{L_{t-}}.
\end{equation}
Another practical consideration is that the predictable function $\tilde\eta$ in \eqref{eq:190925.1} is not finite-valued at the point where $x_2 = - L_-$.
 
To accommodate functions with restricted domain, we say that $$ \text{a predictable function $\xi$ is \emph{compatible} with  $X$ if } \xi(\Delta X) \text{ is finite-valued}, \P\text{--a.s.}$$
It is formally straightforward to apply the formula \eqref{eq:190924.5} to a predictable function $\xi$.  
However, it is \emph{a priori} not clear that the formula will be well-defined;  in particular, for $\xi=\tilde \eta$ in \eqref{eq:190925.1} it is not clear that one obtains
\begin{equation}\label{eq:191003.1} 
\sum_{0<t \leq \cdot} |\xi_t(\Delta X_t) - D \xi_t(0) \Delta X_t|< \infty.
\end{equation}  
Ideally, we would like to have a calculus that gives rise only to those predictable functions $\xi$ for which the integral $\int_0^\cdot\xi_t (\d X_t)$ is always well-defined so we do not have to manually check admissibility every time a new $\xi$ arises. The precise formulation of \emph{two} such subclasses, whose elements will be called \emph{universal representing functions}, is given in Appendix~\ref{A:proofs}. We denote by $\I_{0\R}^{d,n}$ the subset of universal representing functions that map $\R^d$--valued processes to $\R^n$--dimensional processes.  We also set $\I_{0\R}=\bigcup_{d,n\in\N}\I_{0\R}^{d,n}$. For computations involving characteristic functions, it is convenient to generalize the \'Emery formula~\eqref{eq:190924.5} to a subset of complex predictable functions by interpreting $D\xi$ and $D^2\xi$ as complex derivatives. The subset of all universal representing  functions that are complex-differentiable near the origin will be denoted by $\I_{0\Cx}$.

Propositions~\ref{P:integral}--\ref{P:composition} in Appendix~\ref{A:proofs} show that the class $\I_{0\R}$ is self-contained; if we use standard real-valued operations we are guaranteed to stay within $\I_{0\R}$. The class $\I_{0\Cx}$ is also self-contained provided the transformations we apply are complex-differentiable, as will be the case in all our examples involving complex numbers.
These results guarantee that we will only ever encounter functions in $\I_{0\R}$ (resp., $\I_{0\Cx}$) and so will never have to check the conditions of Definition~\ref{D:core} manually.  

\subsection{Integral notation}
Just as the McKean calculus of Subsection~\ref{sect: Brownian}, the simplified stochastic calculus is most intuitive when expressed in differential form, such as \eqref{eq:190924.0} and \eqref{eq:190924.1}. When one wishes to speak of the integrated process whose increments are equal to $\xi_t(\d X_t)$, one typically just introduces a new label, say $Y$, writing $ \d Y_t = \xi_t(\d X_t)$. There is nothing wrong with the relabelling approach; it does deliver all the immediate benefits of the simplified calculus and helps to keep technicalities to a minimum.

Side-by-side with the intuitive differential approach, we want to offer the reader an alternative `high-level' view of the calculus where the roles of $\xi$ and $X$ are acknowledged explicitly. Accordingly, the process with increments $\xi_t(\d X_t)$, starting at 0, will be denoted by $\xi \circ X$, i.e.,
\[
 \xi \circ X = \int_0^\cdot \xi_t(\d X_t).
\]
 The high-level notation may seem a little abstract at first, but it offers distinct benefits such as compactness and flexibility. For example, in the integral notation one can write
\begin{align*}
[X,X] =& \int_0^\cdot (\d X_t)^2 = \id^2 \circ X;\\
[[X,X],X] =& \int_0^\cdot(\d X_t)^3 = \id^3\circ X;\\
[[[X,X],X],X] = [[X,X],[X,X]] =& \int_0^\cdot(\d X_t)^4 = \id^4\circ X.
\end{align*}
Here we let $\id$ denote the identity function; 
$$\id(x) = x.$$ 
Below we use $\id_1(x)=x_1$ for the first component, $\id_2$ for the second, etc., where required.

The notation $\xi\circ X$ also emphasizes the universality of the transformation $X\mapsto \xi\circ X$. In the same way that $[X,X]$ is well-defined for \emph{any} semimartingale $X$, the process $|\id|^\alpha\circ X$ is well-defined for any semimartingale $X$ and any real $\alpha\geq 2$. This brings us to other universal transformations that are commonly used in the literature. For example, provided that $X$ and $X_-$ are different from zero, the literature defines $\Log(X)$ as the process of cumulative percentage change in $X$, i.e.,
$$\d \Log(X)_t = \frac{\d X_t}{X_{t-}}.$$
Thus, in the integral notation formula \eqref{eq:190924.1} reads  
\begin{equation} \label{eq: Log(K/L) circ}
\Log\left(\frac{K}{L}\right) = \left(\frac{1+\id_1}{1+\id_2}-1\right)\circ (\Log(K),\Log(L)).
\end{equation}
Provided that the cumulative percentage changes in $K$ and $L$ are well-defined, formula \eqref{eq: Log(K/L) circ} holds for arbitrary semimartingales $K$ and $L$; it is new in this generality as far as we know.
\section{Simplified stochastic calculus}\label{sect: 2}
\subsection{Composite rules}
The simplified stochastic calculus rests on a sequential application of Propositions~\ref{P:integral}--\ref{P:composition}. In practice, one would not use these propositions directly but instead combine them into composite rules that transform a represented process $Y=Y_0+\xi\circ X$ into another represented process $Z$. In differential notation, we have $\d Y_t = \xi_t (\d X_t)$ and the first two rules read as follows.
\begin{description}[leftmargin=0cm]
\item[Stochastic integration with respect to $Y$] For a locally bounded process $\zeta$  and the integral $Z = \int_0^\cdot \zeta_t \d Y_t$ one has 
$$ \d Z_t = \zeta_t \xi_t (\d X_t).$$
\item[It\^o formula for $Z=f(Y)$] For a suitably smooth function $f$ such that $Y$ and $Y_-$ lie in the interior of the domain of $f$ one has 
$$ \d Z_t = \d f(Y_t) = f(Y_{t-}+\d Y_t)-f(Y_{t-}) = f(Y_{t-}+\xi_t(\d X_t))-f(Y_{t-}).$$
\end{description}
We now restate the above fully rigorously in integral notation.
\begin{corollary}\label{C:191018}
Assume $\xi\in\I_{0\R}^{d,n}$ (resp., $\I_{0\Cx}^{d,n}$) is compatible with $X$ and consider the $n$--dimensional  
process 
$$Y = Y_0 + \xi\circ X.$$
The following rules then apply.
\begin{itemize}
\item \emph{Stochastic integration:} For a locally bounded $\R^{m\times n}$--valued (resp., $\Cx^{m\times n}$--valued) predictable process $\zeta$ we have $\zeta \xi\in\I^{d,m}_{0\R}$ (resp., $\I^{d,m}_{0\Cx}$) and 
\begin{equation}\label{eq: integration rule}
Z = Z_0 + \int_0^\cdot \zeta_u \d Y_u = Z_0 +\zeta \xi\circ X;
\end{equation}
\item \emph{Smooth transformation (`It\^o's formula'):} Assume $Y$ and $Y_-$ remain in an open subset $\mathcal{U}$ of $\R^n$ (resp., $\Cx^n$) where the  function $f:\mathcal{U}\to\R^m$ is twice continuously differentiable (resp. where $f:\mathcal{U}\to\Cx^m$ is analytic). Then $f(Y_{-}+\xi)-f(Y_{-})$ is in $\I_{0\R}^{d,m}$ (resp., $\I_{0\Cx}^{d,m}$) and it is compatible with $X$. Furthermore,
\begin{equation}\label{eq: Ito rule} 
Z = f(Y) = f(Y_0) +(f(Y_{-}+\xi)-f(Y_{-}))\circ X.
\end{equation}
\end{itemize}
\end{corollary}
Let us outline a general scheme of how the two composite rules are applied in practice. At the outset, one will designate the primitive input to the problem at hand; this input process is thereafter labeled $X$. For example, in Subsections~\ref{sect: Brownian} and \ref{sect: first steps} it is natural to start from the bivariate process 
\begin{equation}\label{eq:191018.1}
X=(\Log(K),\Log(L)).
\end{equation} 
Observe that $X$ is always representable with respect to itself thanks to Proposition~\ref{P:integral} with $\zeta$ equal to  the $d\times d$ identity matrix. Starting off from the trivial representation $X =X_0 + \id\circ X$, i.e., taking $Y=X$ and $\xi(x) = x$ in Corollary~\ref{C:191018}, one  applies smooth transformation or stochastic integration as required to obtain the first intermediate result $Z$. This intermediate result (relabeled $Y$) becomes the input to the next application of Corollary~\ref{C:191018} producing the next intermediate output $Z$. The $Y\to Z$ pattern is repeated until one reaches the desired output $Z$; in our example the goal is
\begin{equation}\label{eq:191018.2}
Z = \Log\left(\frac{K}{L}\right).
\end{equation}
Table~\ref{tab:1} illustrates the steps required in the transition from \eqref{eq:191018.1} to \eqref{eq:191018.2}.

\begin{table}[tbp] \centering%
\begin{tabular}{ccccc}
\hline\vspace{-0.25cm}\\
\# & $\d Y$ & \multicolumn{1}{c}{operation $Y\to Z$}& &\multicolumn{1}{c}{$\d Z$} \smallskip\\ \hline \\[-0.2cm]
1 
& 
$\d\!\left[ \!\!\!
\begin{array}{c}
\Log(K)  \\[-1pt]
\Log(L)%
\end{array}%
\!\!\!\right] $ 
& 
\begin{tabular}{c}
integration \\ 
$\zeta =\left[\!\!
\begin{array}{cc}
K_{-} & 0 \\[-1pt] 
0 & L_{-}%
\end{array}
\!\!\right] $ 
\end{tabular}
&
& 
$\d\!\left[\!\!
\begin{array}{c}
K \\[-1pt] 
L%
\end{array}%
\!\!\right] =\left[ \!\!
\begin{array}{cc}
K_{-} & 0 \\[-1pt] 
0 & L_{-}%
\end{array}%
\!\!\right] 
\d\!\left[ \!\!
\begin{array}{c}
\Log(K)  \\[-1pt] 
\Log(L)%
\end{array}%
\!\!\right] $ \smallskip \\ \vspace{-0.2cm}\\ 
2 
& 
$\d\!\left[\!\!
\begin{array}{c}
 K \\[-1pt] 
 L%
\end{array}%
\!\!\right] $ 
&
\begin{tabular}{c}
     smooth transformation \\ 
	    $f(K,L)=\frac{K}{L}$ 
\end{tabular} 
&
& 
\begin{tabular}{rl}
$\d\!\left( \frac{K}{L}\right)$\!\!\! &$=\left( \frac{K_{-}+\d K}{L_{-}+\d L}-\frac{K_{-}}{L_{-}}\right)$\smallskip\\
 &$=\frac{K_{-}}{L_{-}}\left( \frac{1+\d\mathcal{L}\left( K\right) }{1+\d\mathcal{L}\left( L\right) }-1\right) $
\end{tabular} \smallskip \\ \\[-0.2cm]
3 
& 
$\d\!\left(\dfrac{K}{L}\right)$
& 
\begin{tabular}{c}
integration\\ 
$\zeta =\frac{L_{-}}{K_{-}}$
\end{tabular}
&
& 
\begin{tabular}{rl}
$\d\Log\!\left( \frac{K}{L}\right)$\!\!\! &$=\frac{L_{-}}{K_{-}}\,\d\!\left( \frac{K}{L}\right) $\smallskip\\
&$= \frac{1+\d\Log(K)}{1+\d\Log(L)}-1 $
\end{tabular} \bigskip\\ \hline
\end{tabular}%
\bigskip
\caption{Schematic derivation of \eqref{eq: Log(K/L) circ} by means of Corollary~\ref{C:191018}}\label{tab:1}
\end{table} 

The discussion above concerns formal calculations where the rules of the simplified calculus are applied mechanically. Many users will prefer a more intuitive approach whose main idea is apparent in the last column of Table~\ref{tab:1}. Here one observes that the calculus traces the behaviour of jumps and so one may effectively pretend that $X$, $Y$, and $Z$ are finite-variation pure-jump processes. In this way, it is possible to arrive at the correct $\xi$ even without applying formal rules. In the context of \eqref{eq: Log(K/L) circ}, for example, suppose $K$ increases by 50\% and $L$ increases by 20\%. The percentage change in $\sfrac{K}{L}$ is then precisely
$$\frac{1+0.5}{1+0.2}-1 = 25\%. $$
Therefore, formula \eqref{eq: Log(K/L) circ}, among other things, describes jump transformations: every time $\Log(K)$ jumps by $x_1$ (e.g., 0.5) and $\Log(L)$ jumps by $x_2$ (e.g., 0.2) the process $\Log(\sfrac{K}{L})$ jumps by
$$ \xi(x) = \frac{1+x_1}{1+x_2}-1.$$

As a further example, let us see how the rules of the simplified calculus can be used to obtain the representation of the logarithmic return in terms of the rate of return.
\begin{example}[Representation of the log return in terms of the rate of return]\label{E:181127.1}
Let $S>0$ stand for the value of an investment with $S_->0$ and let $X = \Log(S)=\int_0^\cdot \sfrac{\d S_t}{S_{t-}}$ be the cumulative rate of return on this investment. On a purely intuitive level, thinking only of jump transformations, one can write 
\begin{align*}
\d \log S_{t}  = \log \left(S_{t-}+\d S_t\right)-\log S_{t-} = \log \left(1+ \frac{\d S_t}{S_{t-}}\right) =\log \left(1+ \d \Log(S)_t\right).
\end{align*}
More formally, the integration rule yields $S = S_0 + S_-\,\id\circ \Log(S)$ and smooth transformation then gives
\begin{equation*}
\log S = \log(S_0 + S_-\,\id\circ \Log(S)) = \log S_0 + (\log (S_-+S_-\,\id)-\log S_-)\circ \Log(S).
\end{equation*}
Both approaches yield the representation 
\begin{equation*}
\log S = \log S_0 + \log(1+\id)\circ \Log(S)
\end{equation*}
for any semimartingale $S$ such that $S_->0$ and $S>0$.\qed
\end{example}
As the final introductory example consider the representation of quadratic covariation.
\begin{example} [Representation of quadratic covariation]\label{E:180809}
The quadratic covariation $[X,Y]$ satisfies (or, as in \citealp{meyer.76}, is defined by) the identity 
\begin{equation*}
X Y=X_0Y_0+\int_0^\cdot X_{t-}\d Y_t + \int_0^\cdot Y_{t-}\d X_t + [X,Y].
\end{equation*}
This yields
\begin{alignat*}{3}
\d [X,Y]_t =& \qquad\qquad\qquad\d (X_t Y_t)    &   -{}&X_{t-}\d Y_t - Y_{t-}\d X_t&&\\
=&\, (X_{t-}+\d X_t)(Y_{t-}+\d Y_t)-X_{t-}Y_{t-}& \ -{}&X_{t-}\d Y_t - Y_{t-}\d X_t&\ ={}& \d X_t \d Y_t.
\end{alignat*}
More formally, the integration  and smooth transformation rules yield
\[ 
[X,Y]= ((X_-+\id_1)(Y_-+\id_2)-X_-Y_--X_-\,\id_2-Y_-\,\id_1)\circ (X,Y) = \id_1\,\id_2\circ(X,Y).
\]
Thus, in the differential notation one can rigorously write $\d [X,X]_t = (\d X_t)^2$ for any univariate semimartingale $X$. \qed
\end{example}
Section~\ref{sect:190714.1} contains many more explicit representations that are useful in practice. Some of these are well known in the specialist literature, while others are new. Proposition~\ref{P:composition}, in particular,  is a powerful tool for obtaining new representations from old ones. We summarize it here in the form of a composition rule. Observe that in the differential notation the rule is completely natural; it asserts that 
\begin{equation}\label{eq:191219.1}
\text{$\d Y_t = \xi_t(\d X_t)$\quad and \quad$\d Z_t = \psi_t(\d Y_t)$ \qquad \text{yield} \qquad$\d Z_t = \psi_t(\xi_t(\d X_t))$.}
\end{equation}

\begin{corollary}[Composition of representations]\label{C:composition}
Assume $\xi\in\I_{0\R}^{d,n}$ is compatible with \smallskip $X$ and consider the $n$--dimensional process 
$Y = Y_0 + \xi\circ X$.
For $\psi\in\I_{0\R}^{n,m}$ compatible with $Y$ one obtains that $\psi(\xi)\in \I_{0\R}^{d,m}$ is compatible with $X$ and 
\begin{equation}\label{eq:composition}
Z = Z_0 + \psi\circ Y = Z_0 + \psi(\xi)\circ X.
\end{equation}
An analogous statement holds with $\I_{0\Cx}$ in place of $\I_{0\R}$.
\end{corollary}

The composition rule allows the user to store some common calculations and `recycle' them later without having to revisit their detailed derivation. Suppose, for instance, that we are given the evolution of  $(\log K,\log L)$ as the primitive input. Thanks to Corollary~\ref{C:composition}, there is no need to calculate everything afresh all the way from $(\log K,\log L)$ to $\Log(\sfrac{K}{L})$. One only computes the passage from $(\log K,\log L)$ to $(\Log(K),\Log(L))$ which yields (see equation \eqref{eq: Log (e^vx) repre} below)
$$(\Log(K),\Log(L)) = \left(\e^{\id_1}-1,\e^{\id_2}-1\right)\circ (\log K,\log L),$$
while the passage from $(\Log(K),\Log(L))$ to $\Log(\sfrac{K}{L})$ can be recycled from \eqref{eq: Log(K/L) circ}. The two results composed together give
\begin{equation}\label{eq:191202.1}
\Log\left(\frac{K}{L}\right) = \left(\e^{\id_1-\id_2}-1\right)\circ (\log K, \log L).
\end{equation}

In differential notation, 
$$\d K_t =\d \e^{\ln K_t} = \e^{\ln K_{t-}+\d\ln K_t}-\e^{\ln K_{t-}} = K_{t-}(\e^{\d \ln K_t}-1) $$
substituted into \eqref{eq:190924.1} yields
$$\frac{\d \left(\sfrac{K_t}{L_t}\right)}{\sfrac{K_{t-}}{L_{t-}}} 
=\frac{1+\e^{\d \ln K_t}-1}{1+\e^{\d \ln L_t}-1}-1 = \e^{\d \ln K_t - \d \ln L_t}-1,
$$
which is the differential equivalent of formula \eqref{eq:191202.1}.
\subsection{\'Emery formula and drift computation}\label{sect: Emery}
Having mastered the art of representing one process by means of another, we would like to obtain an analogon of \eqref{eq: d(K/L) rigorous1}--\eqref{eq: drift}. Our task is to express the drift of the represented process with the help of the characteristics of the representing process. 

To begin with, we collect the predictable characteristics (i.e., the drift rate, the covariance matrix of the associated Brownian motion, and the L\'evy measure) of the input process $X = (\Log(K),\Log(L))$ in equations \eqref{eq: K jump}--\eqref{eq:191023.2} in the more compact form
\begin{equation}\label{eq:191023.3}
 b^{X[h^1]} = \left[\!\!\begin{array}{c}\mu_K\\\mu_L\end{array}\!\!\right]; \qquad c^X = \left[\!\!\begin{array}{cc}\sigma_K^2&\rho_{KL}\sigma_K\sigma_L\\ \rho_{KL}\sigma_K\sigma_L &\sigma_L^2\end{array}\!\!\right];\qquad F^X = \Pi.
\end{equation}
Here $X[h^1]$ denotes the process $X$ with jumps greater than 1 in absolute value removed,
\begin{equation}\label{eq:X[1]}
X[h^1] = X_0 + (\id_1\,\indicator{|\id_1|\leq 1},\id_2\indicator{|\id_2|\leq 1}) \circ X.
\end{equation}
Observe that this precisely matches the decomposition of jumps appearing in \eqref{eq:191023.1}--\eqref{eq:191023.2} and ensures that the drift of $X[h^1]$ is finite.%
\footnote{In contrast, the drift of $X$ need not be well-defined in general; see also Footnote~\ref{fn:1}.} 
More generally, we will denote by $X[h]$ the process containing the small jumps of $X$ as given by a specific truncation function $h$ and observe that $X[h^1]$ corresponds to the choice
\begin{equation}\label{eq:h1}
h^1=(\id_1\,\indicator{|\id_1|\leq 1},\ldots,\id_d\,\indicator{|\id_d|\leq 1}),
\end{equation} 
where $d$ is the dimension of $X$.
The mechanics of truncation are described in Definition~\ref{D:truncation}.

The reader must be warned that $X[0]$, the continuous part of $X$, is not always well-defined. For this reason, $X[0]$ cannot be universally represented in contrast to $X[h^1]$, whose universal representation appears in \eqref{eq:X[1]}. Nonetheless, the situations where $X[0]$ exists do arise in practice, for example, in the \citet{merton.76} jump-diffusion model. In such models we may write 
\begin{equation}\label{eq:191202.2}
\d X_t = \d X[0]_t + \Delta X_t.
\end{equation}
This, when substituted into \eqref{eq:190924.5}, leads to the simplified expression 
\begin{equation}\label{eq:191202.3}
\hspace{-11em}\xi_t(\d X_t) = D\xi_t(0)\d X[0]_t + \frac{1}{2} \sum_{i,j = 1}^d D^2_{ij} \xi_t(0)\, \d\!\left[X^{(i)},X^{(j)}\right]^c_t  
+\xi_t (\Delta X_t),
\end{equation}
which offers a valuable insight into the nature of the \'Emery formula. We observe that the first two terms of \eqref{eq:191202.3} correspond to the McKean calculus for the continuous part of $X$ while the last term accounts for the jumps in $X$. The two components do not interact and can be treated separately. In the most general situation where the decomposition \eqref{eq:191202.2} does \emph{not} exist, one can make \eqref{eq:191202.3} rigorous by adding small jumps to the first term and subtracting them in the last term to obtain
\begin{equation}\label{eq:181124.3}
\xi_t(\d X_t) = D \xi_t (0) \d X[h]_t+ \frac{1}{2} \sum_{i,j = 1}^d D^2_{ij} \xi_t(0)\, \d\!\left[X^{(i)},X^{(j)}\right]^c_t 
+ \left(\xi_t(\Delta X_t) - D \xi_t(0) h(\Delta X_t) \right).
\end{equation}
The original \'Emery formula \eqref{eq:190924.5} corresponds simply to the case where \emph{all} the jumps of $X$ have been added to the first term and subtracted in the last term of \eqref{eq:191202.3}.

We thus come to understand the \'Emery formula as a \emph{spectrum} of equivalent expressions where one can dial the truncation function $h$ all the way down to $0$ or all the way up to $h(x)=x$. In this sense, the truncation is unimportant -- we can always choose $h$ to suit our needs. In a univariate case, one would thus pick $h=0$ if jumps of $X$ have finite variation as in the Merton jump-diffusion model, failing that, $h=\id$ if the drift of $X$ exists, and finally $h=h^1=\id\indicator{|\id|\leq 1}$ in all remaining cases. In a multivariate case this choice can be performed component-wise. With such choice of $h$, the drift of each contributing term in \eqref{eq:181124.3} is guaranteed to be finite. 

We can now perform the feat previously achieved on a smaller scale in \eqref{eq: d(K/L) rigorous1}--\eqref{eq: drift}. By matching each term in \eqref{eq:181124.3} with its drift contribution, one obtains
\begin{equation}
\label{eq:191023.4}
 b^{\xi\circ X}   =  D\xi(0)b^{X[h]} + \frac{1}{2} \sum_{i,j = 1}^d D^2_{ij} \xi(0)c^X_{ij} 
+ \int_{\R^d} (\xi(x)-D\xi(0)h(x))F^X(\d x).
\end{equation}
Formula~\eqref{eq:191023.4} is proved in Theorem~\ref{T:compensators}.

Specifically, with $\xi$ given in \eqref{eq:190924.3} we obtain
\begin{equation*}
D\xi(0) = [1\ \ -1],\qquad D^2\xi(0) = \left[\begin{array}{cc} 0 &-1 \\ -1& 2 \end{array}\right].
\end{equation*}
For the specific input parameters in \eqref{eq:191023.3} and the corresponding truncation function in \eqref{eq:h1} the drift conversion formula \eqref{eq:191023.4} yields
\begin{equation}\label{eq:191007.1}
\begin{split}
 b^{\Log(\sfrac{K}{L})} ={}& \mu_K-\mu_L - \rho_{KL}\sigma_K\sigma_L+ \sigma_L^2 \\
&{}+ \int_{\R^2} \left(\frac{1+x_1}{1+x_2}-1-\left(x_1\indicator{|x_1|\leq1}-x_2\indicator{|x_2|\leq1}\right)\right)\Pi(\d (x_1, x_2)),
\end{split}
\end{equation}
which is the appropriate generalization of \eqref{eq: drift} provided the integral in \eqref{eq:191007.1} converges.%
\footnote{We might consider, for example, a model where the jumps in $\Log(K)$ and $\Log(L)$ are independent, in which case
$$ \Pi(\d x_1,\d x_2) = \indicator{x_2=0}\Pi^K(\d x_1)+\indicator{x_1=0}\Pi^L(\d x_2),$$
meaning capital and labour do not jump simultaneously. We may take $\Pi^K(\d x_1)$ to be lognormal so that $\int_0^\infty x_1\Pi^K(\d x_1)<\infty$ and $K$ has finite mean. The choice $\Pi^L(\d x_2)=x_2^{-2}\indicator{x_2>0}\d x_2$ then provides an example where $L$ has 
infinite variation (even with $\sigma^L = 0$) and infinite mean while the mean of $K/L$ remains finite (see also Theorem~\ref{T:PII}).
} 
 Formula \eqref{eq: E[K_T/L_T]} continues to hold with this choice of $b$; see Theorems~\ref{T:PII} and \ref{T:compensators} below.

\section{Further examples with drift computation}\label{sect:190714.1}

We will now showcase the strength of process representations such as \eqref{eq: Log(K/L) circ} when it comes to computing drifts. We will do so side-by-side with the classical approach.
Let us therefore start with an $\R$--valued L\'evy process written in the classical notation,
\begin{equation}\label{eq:180723.1}
\begin{split}
X &= X_0 + \int_0^\cdot \alpha\d s+\int_0^\cdot \sigma\d W_s
+\int_0^\cdot \int_{\vert x\vert\leq  1} x\widehat{N}(\d s,\d x)
+\int_0^\cdot \int_{\vert x\vert> 1} xN(\d s,\d x),
\end{split}
\end{equation}
where $N$ is a Poisson jump measure, $\Pi$ the corresponding L\'evy measure, 
$$\widehat{N}(\d t,\d x)=N(\d t,\d x)-\Pi(\d x)\d t$$  
the compensated Poisson jump measure, and $\alpha, \sigma \in\R$.

In the simplified stochastic calculus we will never have to write out the decomposition \eqref{eq:180723.1} in full. Instead we just  note that $X$ is an It\^o semimartingale%
\footnote{A generalization of \eqref{eq:180723.1} where $\alpha$, $\sigma^2$, and $\Pi$ are allowed to be stochastic; 
see Definition~\ref{D:Ito} below.} 
with characteristics
\begin{equation}\label{eq:180726.2}
\left(b^{X[h^1]}=\alpha,c^X=\sigma^2,F^X=\Pi\right).
\end{equation}
The notation of \eqref{eq:180726.2} emphasizes the fact that some expressions below, such as \eqref{eq: kappa^X}, remain valid even if 
$b^X$, $c^X$, and $F^X$ are stochastic.

In the next example, we will find the representation for the cumulative percentage change in $\e^{vX}$ for fixed $v\in\Cx$ and use this to compute the moment generating function of $X$.
\begin{example}[Drift of $\Log(\e^{vX})$ for $v\in\Cx$] \label{ex: Levy-Khintchin}
By integration and smooth transformation we have 
\begin{equation}\label{eq:180726.3}
\d\Log(\e^{vX})_t = \frac{\d\kern0.0833em\e^{vX_t}}{\e^{vX_{t-}}}=\frac{\e^{v(X_{t-}+\d X_t)}-\e^{vX_{t-}}}{\e^{vX_{t-}}}
=\e^{v\hspace{0.1em}\d X_t}-1,
\end{equation} 
or equivalently,
\begin{equation} \label{eq: Log (e^vx) repre}
\Log(\e^{vX})=(\e^{v\hspace{0.1em}\id}-1)\circ X.
\end{equation} 
The representing function is $\xi = \e^{v\hspace{0.1em}\id}-1$ with $\xi'(0)=v$ and $\xi''(0)=v^2$. The corresponding \'Emery formula \eqref{eq:181124.3} reads%
\footnote{A helpful mnemonic device for the \'Emery formula is shown in equation \eqref{eq:191202.3} and the subsequent paragraph.}
\begin{equation}\label{eq: Log(e^vx)}
\hspace{0.1cm}\xi_t(\d X_t) = v\d X[h]_t\ +\frac{1}{2}v^2\d [X,X]_t^c+\left(\e^{v\Delta X_t}-1-vh(\Delta X_t)\right).
\end{equation}
It is valid for \emph{any} semimartingale $X$.

It is now straightforward to compute the drift in \eqref{eq: Log(e^vx)}, provided it exists. Specifically, for  an It\^o semimartingale $X$, \eqref{eq: Log(e^vx)} yields the drift rate of
\begin{align}\label{eq: kappa^X}
\hspace{0.1cm}b^{\Log(\e^{vX})} \ =\ vb^{X[h]}\hspace{0.3cm}\ +\frac{1}{2}v^2 c^X\hspace{1.0cm}+\int_\R\left(\e^{vx}-1-vh(x)\right)F^X(\d x). \hspace{-0.55cm}
\end{align}
If, additionally, $X$ is a L\'evy process as in \eqref{eq:180723.1}--\eqref{eq:180726.2}, we obtain
$$ 
\hspace{0.1cm}b^{\Log(\e^{vX})} \ =\ \alpha v\hspace{0.9cm}\ +\frac{1}{2}\sigma^2 v^2\hspace{1.0cm}
+\int_\R\left(\e^{vx}-1-vx\indicator{\lvert x\rvert\leq 1}\right)\Pi(\d x) \hspace{-0.55cm}
$$
as long as the integral is finite and, in analogy to \eqref{eq: E[K_T]}, 
\begin{align}
\E\left[\e^{v(X_T-X_0)}\right]=\exp\left(b^{\Log(\e^{vX})}T\right).\label{eq:180726.1}
\end{align}
The drift rate $\kappa^X(v)=b^{\Log(\e^{vX})}$ is known as the cumulant function of the r.v. $X_1-X_0$.%
\footnote{\label{foot:mart}Formula \eqref{eq:180726.1} holds thanks to Theorems~\ref{T:PII} and \ref{T:compensators}. Note that 
\eqref{eq:180726.1} is in fact the L\'evy-Khintchin formula applied to the L\'evy process $X$ \citep[Theorem 8.1]{sato.99}.}%
\qed
\end{example}

\begin{remark}
Let us now consider the same calculation using the form \eqref{eq:180723.1}. It\^o's formula \citep[Theorem~4.4.7]{applebaum.09} gives
\begin{equation}\label{eq:180726.5}
\begin{split}
\e^{vX}-\e^{vX_0} &= \int_0^\cdot v\e^{vX_{s-}}(\alpha\d s+\sigma\d W_s)+\frac{1}{2}\int_0^t v^2\e^{vX_{s-}}\sigma^2\d s\\
&\qquad+\int_0^\cdot \int_{\vert x\vert\leq  1} \left(\e^{v(X_{s-}+x)}-\e^{vX_{s-}}\right)\widehat{N}(\d s,\d x) \\
&\qquad+\int_0^\cdot \int_{\vert x\vert>  1} \left(\e^{v(X_{s-}+x)}-\e^{vX_{s-}}\right)N(\d s,\d x)\\
&\qquad+\int_0^\cdot \int_{\vert x\vert\leq  1} \left(\e^{v(X_{s-}+x)}-\e^{vX_{s-}}-v\e^{vX_{s-}}x\right)\Pi(\d x)\d s.
\end{split}
\end{equation}
Integration yields \citep[Section~4.3.3]{applebaum.09}
\begin{equation}\label{eq:180726.6}
\begin{split}
\Log(\e^{vX}) =& \int_0^\cdot \e^{-vX_{s-}}\d \e^{v X_s}\\
=&\int_0^\cdot\left(\alpha v+\frac{1}{2}\sigma^2 v^2\right)\d s+\int_0^\cdot\sigma v\d W_s
+\int_0^\cdot \int_{\vert x\vert\leq  1} (\e^{vx}-1)\widehat{N}(\d s,\d x) \\
&+\int_0^\cdot \int_{\vert x\vert>  1} (\e^{vx}-1)N(\d s,\d x)+\int_0^\cdot \int_{\vert x\vert\leq  1} (\e^{vx}-1-vx)\Pi(\d x)\d s.
\end{split}
\end{equation}
Finally, the drift rate is evaluated by computing the drift of each contributing term in \eqref{eq:180726.6},
\begin{equation*}
\begin{split}
b^{\Log(\e^{vX})} &= \alpha v+\frac{1}{2}\sigma^2 v^2+0+0\\
&\qquad+\int_{\vert x\vert>  1} (\e^{vx}-1)\Pi(\d x)
+\int_{\vert x\vert\leq  1} \left(\e^{vx}-1-vx \right)\Pi(\d x).
\end{split}
\end{equation*}
The calculations \eqref{eq:180726.5}--\eqref{eq:180726.6} become much easier in the approach \eqref{eq:180726.3}--\eqref{eq: Log(e^vx)} because the rules of the simplified stochastic calculus are more compact and easier to remember. 
\qed
\end{remark}
\setstretch{1.1}
The main advantage of the simplified calculus is that one does not have to keep track of the drift, volatility, and jump intensities through intermediate calculations. In the next example we will evaluate all three characteristics of the process $Y=\Log(\e^{vX})$ when $X$ is  an It\^o semimartingale. Having all three characteristics is not necessary for our purposes; this example merely shows that the characteristics are easily recalled at any moment --- if needed.

\begin{example} [Characteristics of a represented It\^o semimartingale] \label{ex: characteristics}
Consider $\xi \in \I_{0\R}^{d,n}\cup\I_{0\Cx}^{d,n}$  compatible with an It\^o semimartingale $X$. 
\begin{enumerate}
\item For the `volatility' of the represented process 
$Y = Y_0 + \xi\circ X$ we obtain
$$ c^Y = D\xi(0)c^X D\xi(0)^\top.$$

\item To compute the drift of $Y[g]$, observe that one naturally obtains 
$$Y[g]=Y_0+g\circ Y,$$ 
for any truncation function $g$ equal to identity on a neighbourhood of zero (see Proposition~\ref{P:universal truncation}). 
The chain rule \eqref{eq:composition} now yields $Y[g]=Y_0+ g(\xi)\circ X$. As $D g(0)$ is  by assumption an identity matrix,  we have 
$D (g\circ\xi)(0) = D\xi(0)$,  $D^2(g\circ\xi) = D^2\xi(0)$, and the desired drift is 
\begin{equation*} 
b^{Y[g]} = D\xi(0)b^{X[h]}+\frac{1}{2} \sum_{i,j=1}^d D^2_{ij}\xi(0)c_{ij}^X + \int_\R \left(g(\xi(x))-D\xi(0) h(x)\right)F^X(\d x).
\end{equation*}

\item Finally, let $G$ be a closed $n$--dimensional set not containing zero. The process $\indicator{\id\in G}\circ Y$ counts the jumps of $Y$ whose size is in $G$; its drift yields the jump arrival intensity $F^Y(G)$. The chain rule \eqref{eq:composition} gives $\indicator{\id\in G}\circ Y = \indicator{\xi\in G}\circ X$. The function $\psi = \indicator{\xi\in G}$ satisfies $D\psi(0)=D^2\psi(0)=0$ which yields
\begin{equation*}
	F^Y(G) = b^{\psi\circ X} = 0 + 0 + \int_\R\indicator{\xi(x)\in G}F^X(\d x).
\end{equation*}
Thus, for each $(\omega,t)$, we recognize $F^Y$ as the image (a.k.a.~push-forward) measure of $F^X$ obtained via the mapping 
$\xi$. 
\end{enumerate}

For concreteness, set $Y=\Log(\e^{vX})=(\e^{v\hspace{0.1em}\id}-1)\circ X$ for some fixed $v \in \Cx$ and take $X$ to be the L\'evy process defined by~\eqref{eq:180723.1}. As $\xi(x)=\e^{vx}-1$, we obtain $\xi'(0)=v$ and $\xi''(0)=v^2$, and 
\begin{align*} b^{Y[h^1]}&=\alpha v+\frac{1}{2}\sigma^2 v^2+
\int_\R \left((\e^{vx}-1)\indicator{\lvert \e^{vx}-1\rvert\leq 1}-vx\indicator{\lvert x\rvert\leq 1}\right)\Pi(\d x);\\
c^Y &= \sigma^2 v^2;\\
F^Y(G) &= \int_\R\indicator{G}(\e^{vx}-1)\Pi(\d x).
\end{align*}
We thus conclude that if $X$ is a L\'evy process then $Y=\Log(\e^{vX})$ is again a L\'evy process for all $v\in\Cx$.
\qed
\end{example}

The next example illustrates the convenience of composing two representations without having to work with their predictable characteristics.

\begin{example}[Maximization of exponential utility] \label{ex: max exp utility}
Fix a time horizon $T > 0$, and assume that $X$ is a one-dimensional L\'evy process given by \eqref{eq:180723.1}--\eqref{eq:180726.2}. Consider an economy con\-si\-sting of one bond with constant price $1$ and of one risky asset with price process $S = \e^X$. Moreover, consider an agent with exponential utility function $u: w \mapsto - \e^{-w}$.

Since $X$ is assumed to have stationary and independent increments and since we consider an exponential utility function, it is reasonable to conjecture that the optimal portfolio is a constant dollar amount $\lambda \in \R$ invested in the risky asset. Denote by $R=\Log(\e^X)$ the cumulative yield on an 1\$ investment in the risky asset. Normalizing initial wealth to zero, the optimal wealth process equals $\lambda R$ and its expected utility is $\E[\e^{-\lambda R_T}]$. In analogy to \eqref{eq: E[K_T]} the expected utility can be obtained via the time rate of the expected percentage change of the quantity 
$\e^{-\lambda R}$.  
This is nothing other than the drift rate of the process $\Log(\e^{-\lambda R})$. Provided this drift, commonly denoted by $\kappa^R(-\lambda)$, is finite, the expected utility will be equal to $\E[\e^{-\lambda R_T}]=\e^{\kappa^{R}(-\lambda)T}$, cf. \eqref{eq:180726.1}.

Formula \eqref{eq:180726.3} and the composition rule~\eqref{eq:191219.1} give $\d R_t=\sfrac{\d \e^{X_t}}{\e^{X_{t-}}} =\e^{\d X_t}-1$ and
\begin{align} \label{eq: Log (e^{-lambda R})}
\d \Log(\e^{-\lambda R})_t=\frac{\d\e^{-\lambda R_{t}}}{\e^{-\lambda R_{t-}}} &= \e^{-\lambda\d R_t}-1 = \e^{-\lambda(\e^{\d X_t}-1)}-1. 
\end{align}
For $\xi(x) = \e^{-\lambda(\e^x-1)}-1$ one has $\xi'(0)=-\lambda$ and $\xi''(0)=\lambda^2 - \lambda$. Hence the desired drift reads
\begin{equation} \label{eq: kappa^R}
	\kappa^{R}(-\lambda) =  b^{\Log(\e^{-\lambda R})}  =  - \alpha\lambda + \frac{\sigma^2}{2}\left(\lambda^2 - \lambda\right) +\int_{\R}\left(
\e^{-\lambda \left( \e^{x}-1\right) }-1+\lambda x\indicator{\lvert x\rvert\leq 1}\right) \Pi(\d x).
\end{equation}
This expresses the cumulant function of $R$ by means of the jump intensity of the process $X$. 

Under the non-restrictive assumptions of \citet[Corollary~3.4]{fujiwara.miyahara.03} the expression \eqref{eq: kappa^R} is finite for all $\lambda\in \R$ and $\lambda\mapsto \kappa^R(-\lambda)$ has a unique maximizer $\lambda_*$ \citep[Proposition~3.3]{fujiwara.miyahara.03}. Under the same assumptions, $R$ is locally bounded and it follows from the results in \citet{biagini.cerny.11} that $\sfrac{\lambda_*}{S_-}$ is the optimal strategy in a sufficiently wide class of admissible strategies for trading in $S$, therefore $\lambda_*$ is the optimal dollar amount to be invested. \qed
\end{example}
\begin{remark} \label{rem: no characteristics}
In \citet{fujiwara.miyahara.03} the previous calculation is performed in two steps: first the characteristics of the yield process $R = \Log(\e^{X})$ are computed and these are then plugged into the L\'evy-Khintchin formula \eqref{eq: kappa^X} of Example~\ref{ex: Levy-Khintchin} to evaluate the cumulant function $\kappa^R(-\lambda)$, which after some cancellations and change of variables gives \eqref{eq: kappa^R}. The two-stage procedure is akin to using $\d t, \d W$ notation which, too, forces the user to keep track of the characteristics at every step of a multistage calculation, see 
\eqref{eq: K}--\eqref{eq: L} 
and \eqref{eq: ItoCondensed}. The simplified stochastic calculus allows us to maintain a model-free formulation until the very end so that the drift calculation is performed only once, when the drift is finally needed. \qed
\end{remark}

The online appendix \cite{cr_online} discusses affine models and the derivation of Riccati equations as an additional example.
\setstretch{1.2} 

\section{Drift under a change of measure} \label{sect:190620} Next we will demonstrate that the simplified calculus becomes very powerful when it comes to evaluating drifts under a different measure.
\begin{example} [Minimal entropy martingale measure] \label{ex: MEMM}
Let us continue in the economic setting of Example~\ref{ex: max exp utility} with the stock price process $S=\e^X$, dollar yield process $R=\Log(\e^X)$, exponential utility $u:w\mapsto \e^{-w}$, and optimal wealth process $\lambda_* R$. Under the assumptions of \citet[Corollary~3.4]{fujiwara.miyahara.03} the Radon-Nikodym derivative $Z_T = \sfrac{\d \Qu }{ \d \P}|_{\sigalg F_T}$ of the representative agent pricing measure  is proportional to the marginal utility evaluated at the optimal wealth, that is, to $\e^{-\lambda_* R_T}$; see \citet[Theorem~3.1 and Corollary~4.4(3)]{fujiwara.miyahara.03}. This $\Qu$ is known in the literature as the minimal entropy martingale measure 
and the corresponding density process $Z$ satisfies $Z_t = \e^{-\lambda_* R_t-\kappa^R(-\lambda_*)t}$ for all $t \in [0,T]$. The process $Z$ is  a true martingale thanks to Corollary~\ref{C:compensatorsPIIQ}. 

\setstretch{1.15}
To value contingent claims on the stock $S = \e^X$, it is necessary to compute the characteristic function of $X$ under $\Qu$.  The required cumulant function $\kappa_\Qu^X(v)$ is just the expected rate of change of $\e^{vX}$ under $\Qu$, i.e., the  $\Qu$--drift rate of 
$\Log(e^{vX})_t$. By Theorem~\ref{T:Girsanov2}\ref{T:Girsanov2:ii} and Example~\ref{E:180809} this $\Qu$--drift  is the same as the $\P$--drift of
\begin{align*}
\d\Log(e^{vX})_t+\d\bigs[\Log(e^{vX}),\Log(\e^{-\lambda_* R})\bigs]_t
&=\d\Log(e^{vX})_t+\d\Log(e^{vX})_t\d\Log(\e^{-\lambda_* R})_t\\ 
&= (\e^{v\d X_t} - 1)\e^{-\lambda_*(\e^{\d X_t}-1)},
\end{align*}
where the second equality combines the representations of $\Log(e^{vX})$ and $\Log(\e^{-\lambda_* R})$ obtained earlier in \eqref{eq:180726.3} and \eqref{eq: Log (e^{-lambda R})}.

The function $\psi(x) = (\e^{vx} - 1)\e^{-\lambda_*(\e^x-1)}$ satisfies 
$\psi'(0)=v$ and $\psi''(0)=v^2-2\lambda_* v$. 
Consequently, if the drift exists, the $\P$--drift rate of $\psi\circ X$ reads
\begin{equation}\label{eq:191106.1}
b^{\psi\circ X}=vb^{X[h]} + \frac{c^X}{2}\left(v^2 - 2\lambda_* v\right) + \int_{\R} \left(\psi(x)  - v h(x)\right) F^X(\d x).
\end{equation}
In the L\'evy setting this yields
\begin{equation}\label{eq:180806.1}
\begin{split}
	\hspace{-0.3cm}\kappa_\Qu^X(v) &= b_\Qu^{\Log(\e^{vX})} = b^{\Log(\e^{vX})+[\Log(\e^{vX}),\Log(\e^{-\lambda_* R})]}\\
	&=\alpha v + \frac{\sigma^2}{2}\left(v^2 - 2\lambda_* v\right) + \int_{\R} \left((e^{vx} - 1) \e^{-\lambda_*(\e^x-1)}  - v x\indicator{\lvert x \rvert\leq 1}  \right) \Pi(\d x),
\end{split}
\end{equation}
whenever the integral on the right-hand side is finite. \qed
\end{example}

\begin{remark}
The standard calculus using the formulation \eqref{eq:180723.1} requires much more work. First, one must find an explicit expression for $\log Z$, which after a significant effort reads
\begin{equation*}
\begin{split}
\log Z =& -\int_0^\cdot \lambda_*\sigma\d W_s-\frac{1}{2}\int_0^\cdot \lambda_*^2\sigma^2\d s+\int_0^\cdot \int_\R -\lambda_*(\e^x-1)\widehat{N}(\d s,\d x)\\
&+\int_0^\cdot\int_\R \left(-\lambda_*(\e^x-1)-\left(\e^{-\lambda_*(\e^x-1)}-1\right)\right)\Pi(\d x)\d s,
\end{split}
\end{equation*}
assuming $\ln Z$ has finite mean. Next, one constructs a new Brownian motion for the measure~$\Qu$,
$$\d W^\Qu_t = \d W_t + \lambda_*\sigma\d t,$$
 and a new compensated Poisson jump measure 
$$\widehat{N}^\Qu(\d t,\d x) = \widehat{N}(\d t,\d x) + \left(1-\e^{-\lambda_*(\e^x-1)}\right)\Pi(\d x)\d t,$$
both using a custom-made formula, see \citet[Theorem ~5.2.12 and Exercise~5.2.14]{applebaum.09} and \citet[Theorem~1.32 and Lemma~1.33]{oksendahl.sulem.07}.

These quantities are then substituted into \eqref{eq:180726.6} to obtain
\begin{equation}\label{eq:180806.2}
\begin{split}
\Log(\e^{vX}) &= \int_0^\cdot \left(\alpha v+\frac{\sigma^2}{2} \left(v^2-2\lambda_* v\right)\right)\d s\\
&\qquad +\int_0^\cdot \sigma v\d W^\Qu_s+\int_0^\cdot\int_\R (\e^{vx}-1)\widehat{N}^\Qu(\d x,\d s) \\
&\qquad +\int_0^\cdot \int_{\vert x\vert\leq  1} \left(\e^{-\lambda_*(\e^x-1)}(\e^{vx}-1)-vx\right)\Pi(\d x)\d s\\
&\qquad +\int_0^\cdot \int_{\vert x\vert>  1}\e^{-\lambda_* (\e^x-1)}(\e^{vx}-1)\Pi(\d x)\d s,
\end{split}
\end{equation}
provided the $\Qu$--drift of $\Log(\e^{vX})$ exists.
The drift is now available by summing up the first, fourth, and fifth term in \eqref{eq:180806.2}. In \eqref{eq:180806.1} the same result is available directly after plugging the specific form $h(x)=x\indicator{\vert x\vert\leq 1}$ and the characteristics \eqref{eq:180726.2} into the formula \eqref{eq:191106.1}. The main difference between the two approaches is that \eqref{eq:191106.1} is more compact and arguably much easier to obtain than \eqref{eq:180806.2}.\qed
\end{remark}
\vspace{-0.2cm}
We conclude this section with a bivariate example that makes use of a non-equivalent change of measure.
\vspace{-0.2cm}\setstretch{1.2}
\begin{example}[An option to exchange one defaultable asset for another]\label{E:181124.1}
Fix $d = 2$ and let $X=(X^{(1)},X^{(2)})$ be an $\R^2$--valued L\'evy martingale with the characteristic triplet
$$\left(\left[\begin{array}{cc} 0 \\ 0\end{array}\right],
\left[\begin{array}{cc}\sigma_1^2 & \sigma_{12} \\ \sigma_{12} & \sigma_2^2\end{array}\right],\Pi\right)$$
relative to the truncation function $h(x)=x$.

Consider next two assets with price dynamics given by stochastic exponentials (see \ref{eq:191119}) as 
\[
	S^{(1)} = S^{(1)}_0 \Exp\bigs(X^{(1)}\bigs) = S^{(1)}_0  + \int_0^\cdot S^{(1)}_{t-}\d X^{(1)}_t; \qquad 
	S^{(2)} = S^{(2)}_0  \Exp\bigs(X^{(2)}\bigs) = S^{(2)}_0  + \int_0^\cdot S^{(2)}_{t-}\d X^{(2)}_t.
\]
In financial economics one interprets $\Exp(X)$ as the value of a closed fund with initial investment of \$1 following a trading strategy whose cumulative rate of return equals $X$. For the sake of simplicity, we assume the existence of some risk-free asset that pays zero interest rate.
We furthermore assume that the L\'evy measure $\Pi$ is supported on $[-1,\infty)\times [-1,\infty)$ meaning both assets can default, perhaps simultaneously.

To value an option to exchange asset $S^{(1)}$ for asset $S^{(2)}$ on a specific date $T$, one must compute the expectation
$$ p = \E\left[\left(S^{(1)}_T-S^{(2)}_T\right)^+\right];$$
see \citet{margrabe.78}. Here the expectation is taken with respect to the valuation measure with the risk-free asset as num\'eraire, which explains why $X$ is a martingale.

Let $\Qu_k$ be the valuation measure with $S^{(k)}$ as a num\'eraire, that is, $\sfrac{\d\Qu_k}{\d\P}=\sfrac{S^{(k)}_T}{S^{(k)}_0}$ for each $k\in \{1,2\}$. Then one obtains an alternative expression for the price of the Margrabe option, namely%
\footnote{Note $\Qu^{(k)}$ is not necessarily equivalent but only absolutely continuous with respect to $\P$ since the event $\{S^{(k)}_T = 0\}$ is allowed to have positive probability under $\P$. Nonetheless, the formula \eqref{eq:181120.4} remains valid because the option payoff is zero whenever $S^{(1)}_T$ is zero. For other applications of defaultable num\'eraires see, for example, \citet{FPRuf} and references therein.} 
\begin{equation} 
 p = S^{(1)}_0  \E^{\Qu_1}\left[\left(1-\frac{S^{(2)}_T}{S^{(1)}_T}\right)^{\!\!+}\,\right].\label{eq:181120.4}
\end{equation}

To evaluate \eqref{eq:181120.4} by integral transform methods one needs to compute the expectation
\begin{equation}\label{eq:181127.1}
 \E^{\Qu_1}\left[\indicator{\left\{S_T^{(2)} >0 \right\}} \left(\frac{S^{(2)}_T}{S^{(1)}_T}\right)^{\!\!v}\,\right]
\end{equation}
for certain values $v\in\Cx$. Let us fix such $v$. In the absence of default (of either asset), the computation of \eqref{eq:181127.1} is achieved by evaluating the expected rate of change of $V=\indicator{\left\{S_T^{(2)} >0 \right\}}\left( \sfrac{S^{(2)}}{S^{(1)}}\right)^{v}$ under the measure $\Qu_1$, in analogy to Example~\ref{ex: MEMM}. This is an easy exercise in the simplified stochastic calculus: for a semimartingale $Y$ with  $Y>0$ and $Y_->0$ one obtains 
$$\frac{\d Y_t^v}{Y_{t-}^v} = \frac{(Y_{t-}+\d Y_t)^v-Y_{t-}^v}{Y_{t-}^v} =\left(1+\frac{\d Y_t}{Y_{t-}}\right)^v-1,$$
which yields
$$\Log(Y^v) =  ((1+\id)^v-1)\circ \Log(Y).$$
Composition with $\Log(\sfrac{S^{(2)}}{S^{(1)}})=\left(\sfrac{(1+\id_2)}{(1+\id_1)}-1\right)\circ (\Log(S^{(1)}),\Log(S^{(2)}))$ then gives
\begin{equation}\label{eq:181124.4}
\Log(V)=\left( \left(\frac{1+\id_2}{1+\id_1}\right) ^{\!\!v}-1\right)\circ X.
\end{equation}
Representation~\eqref{eq:181124.4} together with $\Log\bigs(S^{(1)}\bigs)=X^{(1)} = \id_1\circ X$ yields
\begin{equation}\label{eq:181124.1}
\Log(V)+\bigs[\Log(V),\Log\bigs(S^{(1)}\bigs)\bigs] = (1+\id_1)\left( \left(\frac{1+\id_2}{1+\id_1}\right) ^{\!\!v}-1\right)\circ X.
\end{equation}
Let us denote the  function appearing on the right-hand side of \eqref{eq:181124.1} by $\psi$.
In conclusion, without default (neither $S^{(1)}$ nor $S^{(2)}$ is allowed to hit zero)  one obtains from Corollary~\ref{C:compensatorsPIIQ}\ref{C:compensatorsPIIQ:ii} that
\begin{equation}\label{eq:181127.4}
\E^{\Qu_1}\left[\indicator{\left\{S_T^{(2)} >0 \right\}}\left(\frac{S^{(2)}_T}{S^{(1)}_T}\right)^{\!\!v}\,\right] 
= \left(\frac{S^{(2)}_0}{S^{(1)}_0}\right)^v \exp\left(b^{\psi\circ X}T\right).
\end{equation}

In the presence of default (i.e., if either $S^{(1)}$ or $S^{(2)}$ may hit zero), $V$ is no longer a $\P$--semimartingale. However, 
$$S_\uparrow^{(1)} = S^{(1)}_0 \Exp\bigs(\indicator{\id_1\neq-1}\,\id_1\circ X\bigs)$$
is $\Qu_1$--indistinguishable from $S^{(1)}$ with $S_\uparrow^{(1)}>0$ and  $S_{\uparrow-}^{(1)}>0$, $\P$--a.s. Therefore, the process
\begin{align*}
V_\uparrow=\indicator{\left\{S^{(2)} >0 \right\}}\left(\frac{S^{(2)}}{S_\uparrow^{(1)}}\right)^{\!\!v} 
&= \left(\frac{S^{(2)}_0}{S^{(1)}_0}\right)^{\!\!v}  \Exp\left(-\indicator{\id_2= -1}\circ X\right)
\left(\frac{\Exp\left(\id_2\indicator{\id_2\neq -1}\circ X\right)}{\Exp\left(\id_1\indicator{\id_1\neq -1}\circ X\right)}\right)^{\!\!v}\\
&= \left(\frac{S^{(2)}_0}{S^{(1)}_0}\right)^{\!\!v}  \Exp\left(
\left( \left(\frac{1+\id_2 \indicator{\id_2\neq -1}}{1+\id_1 \indicator{\id_1\neq-1}}\right) ^{\!\!v}\indicator{\id_2\neq -1}-1\right)\circ X
\right)
\end{align*}
is a $\P$--semimartingale $\Qu_1$--indistinguishable from $V$. Corollary~\ref{C:compensatorsPIIQ}\ref{C:compensatorsPIIQ:ii} shows that \eqref{eq:181127.4}  goes through 
with a modified jump transformation function
\begin{equation}\label{eq:181127.5}
\psi\left(x_1,x_2\right) = \left(1+x_1\right)
\left( \indicator{x_2\neq -1}\left(\frac{1+\indicator{x_2\neq -1}x_2}{1+\indicator{x_1\neq-1}x_1}\right) ^{\!\!v}-1\right).
\end{equation}

We now proceed to compute the drift rate $b^{\psi\circ X}$ with $\psi$ in \eqref{eq:181127.5}. To this end, note that
\begin{align*}
D\psi (0)=v\left[ \begin{array}{cc} -1 & 1 \end{array} \right];\qquad 
D^{2}\psi (0)=v(v-1)\left[ \begin{array}{cc} 1 & -1 \\ -1 & 1\end{array}\right].
\end{align*}
Next,   formula~\eqref{eq:191023.4} with $h(x)=x$ for all $x \in \R$ yields
\begin{align*}
b^{\psi \circ X}=&\frac{1}{2}\left( \sigma _{1}^{2}-2\sigma_{12}+\sigma _{2}^{2}\right) v(v-1) \\
&+\int_{\R^{2}}\left( \left( 1+x_1\right) 
\left( \left( \frac{1+x_2}{1+x_1}\right) ^{\!\!v}\indicator{x_2\neq -1}-1\right)
+vx_1-vx_2\right) \Pi \left(\d x_1,\d x_2\right)\\
=&\frac{1}{2}\left( \sigma _{1}^{2}-2\sigma_{12}+\sigma _{2}^{2}\right) v(v-1) - \lambda_2^{\Qu_1} 
+ v\left(\lambda _{2}^{\Qu_{1}}-\lambda _{1}^{\Qu_{2}}\right)\\
&+\int_{(-1, \infty)^2}\left( \left( 1+x_1\right) \left( \left( \frac{1+x_2}{1+x_1}\right) ^{\!\!v}-1\right)
+vx_1-vx_2\right) \Pi \left(\d x_1,\d x_2\right),
\end{align*}
as long as the expectation in \eqref{eq:181127.4} is finite. Here, the coefficient 
$$\lambda_{2}^{\Qu_1} = \int_{\R^2}\left(1+x_1\right)\indicator{x_2=-1}\Pi\left(\d x_1,\d x_2\right)$$ 
signifies the arrival intensity of default of asset $2$ under the probability measure $\Qu_1$ and $\lambda_{1}^{\Qu_2}$ has the converse meaning.%
\footnote{The coefficient $\lambda_2^{\Qu_1}$ is the drift rate of the process $\indicator{\id_2=-1}\circ X$ under $\Qu_1$; 
see Theorem~\ref{T:Girsanov2}\ref{T:Girsanov2:ii}.}
 Observe that without default $\kappa(v)=b^{\psi\circ X}$ can be interpreted as the cumulant function of $\ln \sfrac{S^{(2)}_1}{S^{(1)}_1} - \ln \sfrac{S^{(2)}_0}{S^{(1)}_0}$ under $\Qu_1$.

 For concreteness let us now assume that, in the  absence of default, our model follows a bivariate \citet{merton.76} jump-diffusion. In other words, on the open interval $(-1,\infty)\times (-1,\infty)$, the measure $\Pi$ is a fixed multiple $\lambda\geq 0$ of a push-forward measure  of a bivariate normal distribution with parameters 
$$\left(\left[\begin{array}{cc} m_1 \\ m_2 \end{array}\right],
\left[\begin{array}{cc} s_1^2 & s_{12} \\  s_{12} & s_2^2\end{array}\right]\right)$$
through the mapping $(\e^{\id_1}-1,\e^{\id_2}-1)$.
Once the integrals have been evaluated one obtains 
\begin{align*}
\kappa(v) = b^{\psi \circ X}\ =\ &\frac{1}{2}\left( \sigma _{1}^{2}-2\sigma
_{12}+\sigma _{2}^{2}\right) v(v-1) -\lambda _{2}^{\Qu_{1}}\\
&+v\left(\lambda \left( \e^{m_{1}+\frac{1}{2}s_{11}}-\e^{m_{2}+\frac{1}{2}s_{22}}\right) +\lambda _{2}^{\Qu_{1}}-\lambda _{1}^{\Qu_{2}}\right) 
\\
&+\lambda \e^{(1-v)m_{1}+vm_{2}+\frac{1}{2}(1-v)^{2}s_{11}+v(1-v)s_{12}+\frac{1}{2}v^{2}s_{22}}-\lambda 
\e^{m_{1}+\frac{1}{2}s_{11}}.
\end{align*}

Continuing now with the Fourier transform,  \citet[Lemma~4.1]{hubalek.al.06} yields
\[ (1-x)^{+}=\indicator{x=0}+\indicator{x> 0}\int_{\beta +i\R}g(v)x^{v}\d v, \qquad x \geq 0, \]
where $g (v)=\frac{1}{2\pi i}\frac{1}{v(v-1)}$ and $\beta<0$.
Consequently,  using \eqref{eq:181120.4}, the price of the Margrabe option is given as 
\begin{align*}
\frac{p}{S^{(1)}_0} &=   \Qu_{1}\left[S_{T}^{(2)}
=0\right]+\int_{\beta +i\R}g (v)\E^{\Qu_{1}}\left[ \indicator{\left\{S_T^{(2)} >0 \right\}}
\left(\frac{S_{T}^{(2)}}{S_{T}^{(1)}}\right) ^{\!\!v}\right] \d v \\
&=\e^{\kappa(0)T}+  \int_{\beta +i\R} \left(\frac{S_{0}^{(2)}}{S_{0}^{(1)}}\right) ^{\!\!v}\ \frac{1}{2\pi i}\frac{1}{v(v-1)}
\e^{\kappa(v)T}\d v,
\end{align*}
where $\kappa(0) = - \lambda_2^{\Qu_1}$. The integrals are well-defined and Fubini may be applied in the first equality because 
the function $v \mapsto |g(v)|  \indicator{\{S_T^{(2)} >0 \}} (\sfrac{S_T^{(2)}}{S_T^{(1)}})^{\Re v}$ 
is product-integrable on $\beta + i \R$. 
\qed
\end{example}
\setstretch{1.2}

\section{Jumps at predictable times}\label{sect: predictable}
This section illustrates the simplified stochastic calculus in a discrete-time model. The examples below preserve the independent increments feature of the Brownian and the L\'{e}vy-based examples in Subsection~\ref{sect: Brownian} and Sections~\ref{sect:190714.1} and \ref{sect:190620}. This forces the jumps to occur at fixed times. For simplicity we assume these times can be enumerated in an increasing sequence without an accumulation point. In general, the jumps could occur at all rational times, and if we dropped the independent increments assumption, also at random predictable times. These advanced features are handled in full generality in the companion papers \citet{crII,crIII}.

\begin{example}[Maximization of expected utility]\label{E:utility discrete} Denote by $S>0$ the value of a risky asset and assume the logarithmic price $X=\ln S$ is a discrete-time process (Definition~\ref{D:Ito}) with independent and identically distributed increments. Namely, for each $k\in\N$ we let the distribution of $\Delta X_k$ take three values, $\log 1.1,0$, and $\log 0.9$, 
with probabilities $p_u$, $p_m$, and $p_d$, respectively. With zero risk-free rate the value of a fund 
investing \$1 in the risky asset equals $R = \Log(\e^X)$.

To evaluate the expected utility 
$$\E\left[\e^{-\lambda R_t}\right],\qquad t \geq 0,$$
we recall from Example~\ref{ex: max exp utility} the representation of the cumulative percentage change in $\e^{-\lambda R}$. Specifically, from  \eqref{eq: Log (e^{-lambda R})} one obtains $\Log(\e^{-\lambda R}) =\eta\circ X$ with 
\begin{equation}\label{eq:eta} 
\eta(x) = \e^{-\lambda(\e^x-1)}-1.
\end{equation}
Formulae~\eqref{eq:b3} and \eqref{eq:K_discrete} now yield 
\begin{equation}\label{eq:181005.2}
\begin{split}
\E\left[\e^{-\lambda R_t}\right] &= \prod_{k=1}^{\lfloor t\rfloor} \E\left[ 1+\eta_k(\Delta X_k)\right]\\
&=\prod_{k=1}^{\lfloor t\rfloor}  \E\left[ \e^{-\lambda(\e^{\Delta X_k}-1)}\right]
= \left(p_u \e^{-0.1\lambda}+p_m+p_d \e^{0.1\lambda}\right)^{\lfloor t\rfloor},
\end{split}
\end{equation}
for all $t\geq 0$. \qed
\end{example}
\begin{example}[Minimal entropy martingale measure]
We now  compute the minimal entropy martingale measure in the setting of the previous example for some given time horizon $T > 0$. 
Optimizing \eqref{eq:181005.2} over $\lambda$, one obtains an explicit expression for the optimal dollar amount in the risky asset, namely
$$\lambda_* = \frac{\log(\sfrac{p_u}{p_d})}{0.2}.$$ 
Then the random variable $\sfrac{\e^{-\lambda_*R_T}}{\E[\e^{-\lambda_*R_T}]}$ gives the density $\sfrac{\d\Qu}{\d\P}$ of the minimal entropy martingale measure $\Qu$. As in Example~\ref{ex: MEMM}, we seek the expected value of $\e^{vX_{t}}$ under $\Qu$, for fixed $t \in [0,T]$ and $v\in\Cx$, provided the expectation is finite. Because $\Log(\e^{vX})=\xi\circ X$ with 
$$\xi(x)=\e^{vx}-1,$$ 
the desired expectation is given by Corollary~\ref{C:compensatorsPIIQ}\ref{C:compensatorsPIIQ:i} with $\eta$ from \eqref{eq:eta} as follows,  
\begin{align*}
\E^\Qu\left[\e^{vX_t}\right]&=\prod_{k=1}^{\lfloor t\rfloor}
\frac{\E\left[(1+\eta_k(\Delta X_k)) (1+\xi_k(\Delta X_k)) \right]}{\E\left[1+\eta_k(\Delta X_k)\right]}\\
&=\prod_{k=1}^{\lfloor t\rfloor}
\frac{\E\left[\e^{-\lambda_*(\e^{\Delta X_k}-1)}\e^{v\Delta X_k}\right]}{\E\left[\e^{-\lambda_*(\e^{\Delta X_k}-1)}\right]}
=\left(\frac{(1.1^v+0.9^v)\sqrt{p_u p_d}+p_m}{2\sqrt{p_up_d}+p_m}\right)^{\lfloor t\rfloor}\!\!\!\!\!,\quad t \geq 0.
\end{align*}
Here $\E^\Qu[\e^{vX_t}]$ considered as a function of $v\in\Cx$ gives the moment generating function of $X_t=\log S_t$ under $\Qu$ for each $t \geq 0$ and  can be therefore used to price contingent claims by integral transform methods.
\qed
\end{example}

\section{Concluding remarks}\label{sect: conclusions}
\setstretch{1.25}
In this paper we have introduced the notion of `$X$--representation' to describe a generic modelling situation where one starts from a (multivariate) process $X$ whose predictable $\P$--cha\-ra\-cte\-ri\-stics 
are given as the primitive input to the problem. The process $X$, which is trivially representable, is transformed by several applications of the composite rules~\eqref{eq: integration rule}--\eqref{eq: Ito rule} to another process $Y$ which is also $X$--representable. In many situations the required end product is the $\P$--drift of $Y$. These examples include i)\ the construction of partial integro-differential equations from martingale criteria \cite[e.g.,][Theorem 3.3]{vecer.xu.04}; ii) the computation of exponential compensators \citep[e.g.,][Proposition 11.2]{duffie.al.03}; iii) the formulation of optimality conditions for various dynamic optimization problems \citep[e.g.,][Theorem 3.1(v)]{oksendahl.sulem.07}. 

Existing methods force us to keep track of the characteristics (drift, volatility, and jump intensities) throughout all intermediate calculations \citep[e.g.,][Theorem 1.14]{oksendahl.sulem.07}. One of the drawbacks of describing processes via their characteristic triplets  is that the drift and the jump intensities are measure-dependent and the drift additionally also depends on the truncation function $h$.  The new calculus, in contrast, works with $X$--representations, which themselves do not depend on the characteristics in an overt way. This makes individual steps such as change of variables much simpler and the overall calculus more transparent and easy to use. An $X$--representation is converted into a drift only when the drift is really needed. 

The proposed calculus emphasizes the universal nature of transformations such as stochastic integration or change of variables, which can typically be applied in the same way to any starting process $X$. For example, the conversion from the rate of return 
$\sfrac{\d X_t}{X_{t-}}$ to the logarithmic return $\d \ln X_t$ always takes the form $\d \ln X_t = \ln \left(1+\sfrac{\d X_t}{X_{t-}}\right)$. Robust results such as this are helpful in two ways. They offer an easy way to visualize fundamental relationships and separate what is fundamental from what is model-specific. Secondly, they open an avenue for studying richer models where, say, a Brownian motion is replaced with a more general process with independent increments. In the proposed calculus this is possible without additional overheads as long as the Markovian structure of the problem remains unchanged.

Further advantages of the new calculus become apparent when the drift of $Y$ is to be computed under some new probability measure $\Qu$ absolutely continuous with respect to $\P$. The need to switch measures comes particularly from mathematical finance as illustrated in Examples~\ref{ex: MEMM} and \ref{E:181124.1}, but it also arises in natural sciences as part of filtering theory \citep[and the references therein]{sarkka.sottinen.08} and in Monte Carlo simulations \citep[Section~5.4.2]{grigoriu.02}. 
In existing approaches a change of measure requires a custom-made formula that even depends on the form in which the density process $M$ of 
$\sfrac{\d \Qu}{\d \P}$ is supplied. 
If $M$ is written as a stochastic exponential we need one formula, if it appears as an ordinary exponential we need another formula. These formulae convey little intuition and are consequently hard to memorize. 
In the new calculus there is no need to refer to a formula: we simply notice that by Girsanov's theorem the $\Qu$--drift of $V$ equals the 
$\P$--drift of $V+[V,\Log(M)]$. Since it is easy to write down the representation of $V+[V,\Log(M)]$, the Girsanov computation comes at virtually no extra cost.

Somewhat surprisingly, the simplified calculus implies that one can perform classical It\^o calculus on continuous processes by tracing the behaviour of a hypothetical pure-jump finite variation process. While this observation may seem paradoxical at first sight, we believe the emphasis on jumps makes the simplified stochastic calculus less intellectually taxing than classical approaches firmly rooted in Brownian motion.\medskip

\setstretch{1.1}

\setstretch{1.2}

\appendix
\section{Notation and details about the representations} \label{A:proofs}
In this appendix we provide the setup of this paper and the proofs of the statements in Section~\ref{sect: intro}.
Unless specified otherwise, $d$, $m$, and $n$ are positive integers. The underlying filtered probability space is denoted by  
$(\Omega, \sigalg{F}, \filt{F}, \P)$.  Complex integral of a locally bounded $\Cx^n$--valued process $\zeta=\zeta'+i\zeta''$ 
with respect to a $\Cx^n$--valued semimartingale $X=X'+iX''$ is the $\Cx$--valued semimartingale
\begin{align*}
\int_0^\cdot \zeta_t \d X_t &= \int_0^\cdot \zeta'_t \d X'_t -  \int_0^\cdot \zeta''_t  \d X''_t
+i\left(\int_0^\cdot \zeta''_t \, \d X'_t +  \int_0^\cdot \zeta'_t \d X''_t\right).
\end{align*}

We write  $\Cinf^n =\Cx^n \cup \{\NaN\}$ for some `non-number' $\NaN \notin \bigcup_{n\in\N}\Cx^n$ and  
$\widebar{\Omega}^n_\Cx = \Omega \times [0,\infty) \times \Cinf^n$. The symbols $\Rinf$ and $\widebar{\Omega}^n_\R$ have an analogous meaning.
For a predictable function $\xi$ we shall always assume that $\xi(\NaN) = \NaN$. 
If  $\psi: \widebar{\Omega}^n_\Cx \rightarrow \Cinf^m$, with $m \in \N$, denotes another predictable function we shall write $\psi \circ \xi$ or $\psi(\xi)$ to denote the predictable function $(\omega, t, x) \mapsto \psi(\omega, t, \xi(\omega, t, x))$ and likewise with $\Cx$ replaced by $\R$.

Provided they exist, we write $D \xi$ and $D^2 \xi$ for the complex derivatives of 
$\xi:\widebar{\Omega}^d_\Cx \rightarrow \widebar{\Cx}^n$, resp., the real derivatives of $\xi:\widebar{\Omega}^d_\R \rightarrow \widebar{\R}^n$. Note that $D \xi$ has dimension $n \times d$ and $D^2 \xi$ has dimension $n \times d \times d$.
\setstretch{1.2}

\begin{definition}[Two subclasses of universal representing functions]  \label{D:core}
Let $\mathfrak I^{d,n}_{0\Cx}$ denote the set of all predictable functions $\xi: \widebar{\Omega}^d_{\Cx} \rightarrow \Cinf^n$ such that the following properties hold:
\begin{enumerate}[label={\rm(\arabic{*})}, ref={\rm(\arabic{*})}]
			\item\label{I0:i}   $\xi(\omega, t, 0) = 0$ for all $(\omega, t) \in \Omega \times [0,\infty)$.
			\item\label{I0:ii} There is a predictable process $R$ locally bounded away from zero, i.e., with strictly po\-si\-ti\-ve running infimum $R^*$, such that 
			\begin{enumerate}[label={\rm(\alph{*})}, ref={\rm(\alph{*})}]
				\item\label{I0:ii:a} $x \mapsto \xi(\omega, t, x)$ is analytic on $|x|\leq R(\omega,t)$, for all $(\omega, t) \in \Omega \times [0,\infty)$;
				\item\label{I0:ii:b} $\sup_{|x|\leq R} \bigs|D^2 \xi(x)\bigs|$ is locally bounded.
			\end{enumerate}
			\item\label{I0:iii} $D \xi(0)$ is locally bounded.
	\end{enumerate}
We write $\I_{0\Cx}=\bigcup_{k,r\in\N}\I^{k,r}_{0\Cx}$. The subclass $\I_{0\R}$ of predictable functions $\xi:\widebar{\Omega}^d_\R \rightarrow \widebar{\R}^n$ is defined by replacing (a) with the requirement (a') $x \mapsto \xi(\omega, t, x)$ is twice differentiable for $|x|\leq R(\omega,t)$, for all $(\omega, t) \in \Omega \times [0,\infty)$. 
\qed
\end{definition}

Let us provide some context to the previous definition. Most of the time, we are interested  in `real-valued' transformations  $\xi$, which map an $\R^d$--valued semimartingale to an $\R^n$--valued semimartingale. The core class $\I_{0\R}$ is perfectly suited for this purpose.  
The \'Emery formula, as stated in \eqref{eq:190924.5}, works also for complex-valued $\xi$ if we interpret $D\xi(0)$ and $D^2\xi(0)$ as complex derivatives. Such generalization of the theory from $\R$ to $\Cx$, albeit limited by forcing $\xi$ to be analytic at $0$, is helpful when computing characteristic functions, for example. This leads to the definition of  $\I_{0\Cx}$, which is now a proper subclass within a larger core class $\I_0$ of complex-valued universal representing functions nesting also $\I_{0\R}$.

We do not define $\I_0$ itself in this paper but it can be shown that $\I^{d,n}_0$ has a one-to-one correspondence with $\I^{2d,2n}_{0\R}$.  A more general \'Emery formula is available for $\I_0$ but will not be needed in this paper. We  hence refer the interested reader to \citet{crII} for more details. All computations in this paper can be performed either within $\I_{0\R}$ or within $\I_{0\Cx}$ and the two classes are never used jointly.  The arguments for the two cases are often identical but should be read and understood as two separate arguments because the meaning of $D\xi$ is different in the two cases. The proofs for each of the two classes are self-contained.

Let us now briefly show that all terms in  \eqref{eq:190924.5} are well-defined.
\begin{lemma}
\label{L:well-defined}
If $\xi \in \I_{0\R}\cup\I_{0\Cx}$ then the integrals $\int_0^\cdot D\xi_t(0)\d X_t$ and 
$$\int_0^\cdot \sum_{i,j = 1}^d D^2_{ij} \xi_t(0)\, \d\big[X^{(i)},X^{(j)}\big]^c_t$$ 
are well-defined. If, additionally, $\xi$ is compatible with $X$, then \eqref{eq:191003.1} holds.
\end{lemma}
\begin{proof} 
Because $D\xi(0)$ and $D^2\xi(0)$ are locally bounded, the two integrals are well-defined by \citet[Theorem~I.4.31]{js.03}. By assumption, $(\tau_n)_{n \in \N}$ given by $\tau_n=\inf\{t:R^*_t\le \sfrac{1}{n}\}$ is a localizing sequence. Next, let $(\rho_n)_{n \in \N}$ be the localizing sequence from Definition~\ref{D:core}\ref{I0:ii}\ref{I0:ii:b}. Then $(\tau_n\wedge\rho_n)_{n \in \N}$ is again a localizing sequence such that, after localization, $|\xi(x)-D\xi(0)x|\le K|x|^2$ for all $|x|\le \delta$ for some constants $K>0$ and $\delta>0$. This yields, after localization,
\begin{align*} \sum_{0<t \leq \cdot}& |\xi_t(\Delta X_t) - D \xi_t(0) \Delta X_t|
=\!\!\!\!\sum_{\stackrel{0<t \leq \cdot}{|\Delta X_t|\le \delta}}\!\!\!\! |\xi_t(\Delta X_t) - D \xi_t(0) \Delta X_t|
+\!\!\!\! \sum_{\stackrel{0<t \leq \cdot}{|\Delta X_t| > \delta}}\!\!\!\! |\xi_t(\Delta X_t) - D \xi_t(0) \Delta X_t|
<\infty
\end{align*}
as the last sum has only finitely many summands.
\end{proof}

We are now ready to state and prove the main properties of semimartingale representations.
\begin{proposition}[Representation of stochastic integrals]\label{P:integral}
Let $\zeta$ be a locally bounded predictable $\R^{n\times d}$--valued (resp., $\Cx^{n\times d}$--valued) process. Then the predictable function $\xi= \zeta\,\id$ belongs to $\I_{0\R}^{d,n}$~(resp., $\I_{0\Cx}^{d,n}$) and for any $\R^d$--valued (resp., $\Cx^d$--valued) semimartingale $X$ one has
\begin{equation}\label{eq:191006}
\int_0^\cdot \zeta_t \d X_t  = (\zeta\, \id) \circ X.
\end{equation}
\end{proposition}
\begin{proof}
We start with the complex-valued case. As $D\xi=\zeta$ and $D^2\xi=0$
, we have that $\xi$ belongs to $\I_{0\Cx}$ and is compatible with any  $\Cx^d$--valued semimartingale $X$. The \'Emery formula \eqref{eq:190924.5} now yields \eqref{eq:191006}. The real-valued proof proceeds analogously.
\end{proof}

\begin{proposition}[Representation of smooth transformations]\label{P:Ito}
Let $\U \subset \R^d$ (resp., $\U \subset \Cx^d$) be an open set such that $X_-, X \in \U$, let $f: \U \rightarrow \R^n$ be a twice continuously differentiable function (resp., let $f: \U \rightarrow \Cx^n$ be an analytic function), and let 
\begin{align*}
	 \xi^{f,X}(x) = 
		\begin{cases}
			f\left(X_- + x\right) - f\left(X_-\right), &\quad X_- + x \in \U\\
			\NaN, &\quad X_- + x \notin \U
		\end{cases},
		\quad x \in \R^d\ \ (\text{resp., } x\in\Cx^d).
\end{align*}
Then $\xi^{f,X}\in\I_{0\R}^{d,n}$ (resp. $\xi^{f,X}\in\I_{0\Cx}^{d,n}$) is compatible with $X$ and 
$$f(X) = f(X_0)+\xi^{f,X} \circ X.$$
\end{proposition}
\begin{proof}
The first part of the proof is identical for both cases. Note that 
$D \xi^{f,X}(0) = D f(X_-)$ and $D^2 \xi^{f,X}(0) = D^2 f(X_-)$. As both $D f(X_-)$ and $D^2 f(X_-)$ are finite-valued predictable processes, they are locally bounded by \citet[Proposition~3.2]{larsson.ruf.20}. Next, denote by $R \in (0, 1]$ the minimum of 1 and  half of the distance from $X_-$ to the boundary of $\mathcal{U}$ and by $R^*$ its running infimum. The left-continuity of $R$ now yields $R^*>0$. Next, observe that 
\[
	\tau_n=\inf\left\{t \geq 0:R^*_{t}< \frac{1}{n}\right\} \wedge \inf\{t\geq 0: | X_{t-}| > n\}, \qquad n\in\N,
\]
is a localizing sequence of stopping times that makes  
$ \sup_{|x|\leq R} \bigs|D^2 \xi(x)\bigs|$
locally bounded, yielding $\xi^{f,X}\in\I_{0\R}$ (resp., $\I_{0\Cx}$).  As $\xi^{f,X}(\Delta X) = f(X)-f(X_-)$, 
we have $\xi^{f,X}$ is compatible with $X$. 

For $\xi\in\I_{0\R}$, Lemma~\ref{L:well-defined} and  the \'Emery formula \eqref{eq:190924.5} now yield that $f(X_0)+\xi^{f,X} \circ X$ is the  It\^o-Meyer change of variables formula \citep[I.4.57]{js.03} and hence equal to $f(X)$. For $\xi\in\I_{0\Cx}$ the result follows by identifying $\Cx^d$ with $\R^{2d}$ and using the real-valued statement we have just proved.
\end{proof}

\begin{proposition}[Composition of universal representing functions]\label{P:composition}
The space $\I_{0\R}$ is closed under dimensionally correct composition, i.e., if $\xi \in\I_{0\R}^{d,n}$ and $\psi\in\I_{0\R}^{n,m}$ then $\psi\circ\xi\in\I_{0\R}^{d,m}$. An analogous statement holds for $\I_{0\Cx}$.
\end{proposition}
\begin{proof}
The proof is identical for both cases. By localization we may assume that $D \psi(0)$ is bounded and that there exists a constant $\delta_{\psi }>0$
such that $\sup_{|y| \leq \delta_\psi} D^2 \psi(y)$ and consequently also $\sup_{|y| \leq \delta_\psi}  D \psi(y)$ are bounded.
By the same construction, we may assume that there exists a constant $\delta_{\xi }>0$
such that $\sup_{|x| \leq \delta_\xi}  D^2 \xi(x)$ and $\sup_{|x| \leq \delta_\xi}  D \xi(x)$ are bounded. Moreover, there exists also $\delta_{\psi\circ\xi} \in (0, \delta_\xi)$ such that $\sup_{|x| \leq \delta_{\psi\circ\xi}}  \xi(x) < \delta_\psi$.

By direct computation, for all $|x| \leq \delta_{\psi\circ\xi}$ we now have 
\begin{align*} 
D (\psi \circ \xi) (0) &= \sum_{k=1}^n D_k \psi(0) D \xi^{(k)}(0);\\
D^2(\psi \circ \xi) (x) &= \sum_{k,l = 1}^n D^2_{kl} \psi(\xi(x))   D \xi^{(k)}(x)^\top D\xi^{(l)}(x)  
		+ \sum_{k = 1}^n D_k\psi(\xi(x)) D^2\xi^{(k)}(x).
\end{align*}
This yields a positive non-increasing sequence $\left(\delta_{\psi\circ\xi}^{(n)}\right)_{n \in \N}$ and a localizing sequence $(\tau_n)_{n \in \N}$ of stopping times  such that $D (\psi \circ \xi) (0)$ and 
$$ \sup_{|x| \leq \delta^{(n)}_{\psi\circ\xi}}D^2(\psi \circ \xi) (x)$$
are bounded on the stochastic interval $\lc\tau_{n-1},\tau_n\lc$ for each $n\in\N$. The desired process $R_{\psi\circ\xi}$ is obtained by setting $\sum_{n\in\N}\delta_{\psi\circ\xi}^{(n)}\indicator{\lc\tau_{n-1},\tau_n\lc}$.
\end{proof}

\setstretch{1.2}
\section{Truncation and predictable compensators}\label{B}
In this appendix we complement the observations in Subsection~\ref{sect: Emery}. We begin by formally introducing truncation functions. 
\begin{definition}[Truncation function for $X$]\label{D:truncation}
We say that a predictable function $h$ is a truncation function for a 
semimartingale $X$ if $h$ is time-constant and deterministic, compatible with $X$, $ \sum_{0<t \leq \cdot} |\Delta X_t - h(\Delta X_t)|< \infty $, 
and if
$$ X[h] = X - \sum_{0<t \leq \cdot} (\Delta X_t - h(\Delta X_t))$$ 
is a special semimartingale, i.e., if $X[h]$ can be decomposed into the sum of a local martingale and a predictable process of finite variation. \qed
\end{definition}
\begin{proposition}[Universal truncation functions]\label{P:universal truncation}
If a bounded time-constant deterministic function $h$  equals identity on an open neighbourhood of 0 then it is a truncation function for any compatible semimartingale $X$. Furthermore, one then has $h\in\I_{0\R}\cup\I_{0\Cx}$ and 
$$ X[h] = X_0 + h\circ X.$$
\end{proposition}

\begin{proof}
Clearly $h\in\I_{0\R}\cup\I_{0\Cx}$. Lemma~\ref{L:well-defined} with $\xi=h$ now yields 
$ \sum_{0<t \leq \cdot} |h(\Delta X_t) - \Delta X_t |< \infty $.
Next, observe that $\Delta X[h]_t = h(\Delta X_t)$ is bounded, therefore $X[h]$ is special by \citet[I.4.24]{js.03}. Finally, the \'Emery formula \eqref{eq:190924.5}  yields
$$ X_0 + h\circ X = X  + \sum_{0<t \leq \cdot} (h(\Delta X_t) - \Delta X_t ) = X[h],$$
which completes the proof.
\end{proof}

\begin{proposition}[\'Emery formula with truncation]\label{P:Emery truncation}
Assume $\xi\in\I_{0\R}$ (resp., $\I_{0\Cx}$) is compatible with $X$ and let $h$ be a truncation function for $X$. Then 
$$ \sum_{0<t \leq \cdot} \left|\xi_t(\Delta X_t) - D \xi_t(0) h(\Delta X_t)\right|<\infty$$
and 
\begin{equation}\label{eq:191016.1}
\begin{split}
\xi\circ X = \int_0^\cdot D\xi_t(0)\d X[h]_t 
&+ \frac{1}{2}\int_0^\cdot \sum_{i,j = 1}^d D^2_{ij} \xi_t(0)\, \d\bigs[X^{(i)},X^{(j)}\bigs]^c_t \\
								&+ \sum_{0<t \leq \cdot} \left(\xi_t(\Delta X_t) - D \xi_t(0) h(\Delta X_t)\right). 
\end{split}			
\end{equation}
\end{proposition}

\begin{proof}
First, the triangle inequality gives
\begin{align*} \sum_{0<t \leq \cdot} \left|\xi_t(\Delta X_t) - D \xi_t(0) h(\Delta X_t)\right|
\leq& \sum_{0<t \leq \cdot} \left|\xi_t(\Delta X_t) - D \xi_t(0)\Delta X_t\right|\\
&+\sum_{0<t \leq \cdot} \left|D \xi_t(0)\Delta X_t - D \xi_t(0) h(\Delta X_t)\right|\\
<&\infty, \qquad t \geq 0.
\end{align*}
Here the second sum is finite thanks to Lemma~\ref{L:well-defined} and the third due to the local boundedness of $D\xi(0)$ and Definition~\ref{D:truncation}. The identity
$$\int_0^\cdot D\xi_t(0)\d X_t =
\int_0^\cdot D\xi_t(0)\d X[h]_t +
\sum_{0<t \leq \cdot} \left(D \xi_t(0)\Delta X_t - D \xi_t(0) h(\Delta X_t)\right)$$
and the \'Emery formula \eqref{eq:190924.5} now yield the second part of the claim.
\end{proof}

We now introduce notation dealing with predictable compensators. If $X$ is a special semimartingale, we denote by $B^X$ its predictable compensator, i.e., the unique predictable finite variation process starting at zero such that $X-B^X$ is a local $\P$--martingale. 
If $\Qu$ is another probability measure absolutely continuous with respect to $\P$ and $X$ is $\Qu$--special, we denote the corresponding 
$\Qu$--compensator by $B^X_\Qu$.  Finally, we denote by $\nu^X$ the predictable $\P$--compensator of the jumps of $X$, i.e.,  for any compact interval $J\subset \R^d$\ (resp., $\Cx^d$) not containing the origin, $\nu^X([0, \cdot] \times J)$ is the predictable compensator of the finite variation process $ \sum_{0<t\leq\cdot}\indicator{\{\Delta X_t\in J\}}$. 

We shall say that a semimartingale is PII if it has independent increments.  The following result for PII semimartingales  relates drifts to expected values, and hence shall be very useful.  It is proved in \citet[Proposition~2.14 and Theorem~4.1]{crIII}. At this point, we remind the reader that the stochastic exponential $\Exp(X)$ of a one-dimensional semimartingale $X$ is given as the (unique) solution of the stochastic differential equation
\begin{align} \label{eq:191119}
	\Exp(X) = 1 + \int_0^\cdot \Exp(X)_{t-} \d X_t.
\end{align}  

\begin{theorem}\label{T:PII}
Assume $\xi\in\I_{0\R}\cup\I_{0\Cx}$ is compatible with $X$. If $\xi$ is deterministic and $X$ is PII, then $\xi\circ X$, too, is PII. Furthermore, if $\xi\circ X$ is special one has
\begin{alignat}{2}
\E[(\xi\circ X)_t] &= B^{\xi\circ X}_t,& \qquad t &\geq 0;   \label{eq:b2}\\
\E[\Exp(\xi\circ X)_t] &= \Exp\bigs(B^{\xi\circ X}\bigs)_t,& \qquad t &\geq 0. \label{eq:b3}
\end{alignat}
\end{theorem}

Below we evaluate the right-hand-side of \eqref{eq:b2} and \eqref{eq:b3} for two important classes of stochastic processes.
\begin{definition}\label{D:Ito}
We say that $X$ is a discrete-time process if $X$ is constant on $[k-1, k)$ for each $k \in \N$. We say that $X$ is an It\^o semimartingale if for all truncation functions $h$ for $X$ there exists a triplet $(b^{X[h]}, c^X, F^X)$ of predictable processes such that 
$B^{X[h]} = \int_0^\cdot b^{X[h]} \d t$, $[X,X]^c = \int_0^\cdot c^X \d t$, and $\nu^X$ can be written in disintegrated form as 
$\nu^X = \int_0^\cdot \int F^X(\d x) \d t$. 
\qed
\end{definition}

\begin{theorem}\label{T:compensators}
Let $X$ be a semimartingale and let $h$ be a truncation function for $X$. Assume that $\xi\in\I_{0\R}\cup\I_{0\Cx}$ is compatible with $X$ and that $\xi\circ X$ is special. The following statements then hold. 
\begin{enumerate}[label={\rm(\roman{*})}, ref={\rm(\roman{*})}] 
\item\label{T:compensators.ii} If $X$ is a discrete-time process then $\xi\circ X$ is a discrete-time process and
\begin{alignat}{2}
B_t^{\xi\circ X}       &= \sum_{k=1}^{\lfloor t\rfloor}\E_{k-}[\xi_{k}(\Delta X_k)],& \qquad t &\geq 0; \nonumber\\
\Exp(B^{\xi\circ X})_t &= \prod_{k=1}^{\lfloor t\rfloor} \E_{k-}\left[1+\xi_k(\Delta X_k)\right],& \qquad t &\geq 0. \label{eq:K_discrete}
\end{alignat}
\item\label{T:compensators.i} If $X$ is an It\^o semimartingale then $\xi\circ X$ is an It\^o semimartingale and
\begin{align*}
b^{\xi\circ X} &= D\xi(0) b^{X[h]} + \frac{1}{2} \sum_{i,j = 1}^d D^2_{ij} \xi(0)c^X_{ij} 
+ \int_{\R^d}\left(\xi(x) - D \xi(0) h(x)\right)F^X(\d x);\\
B^{\xi\circ X} &= \int_0^\cdot b_t^{\xi\circ X}\d t;\qquad \Exp(B^{\xi\circ X}) = \exp\left(\int_0^\cdot b_t^{\xi\circ X}\d t\right).
\end{align*}
\end{enumerate}
\end{theorem}

\begin{proof}
By \eqref{eq:191016.1}, we have
\[
	\begin{split}
B^{\xi\circ X} = \int_0^\cdot D\xi_t(0)\d B^{X[h]}_t 
&+ \frac{1}{2}\int_0^\cdot \sum_{i,j = 1}^d D^2_{ij} \xi_t(0)\, \d\bigs[X^{(i)},X^{(j)}\bigs]^c_t \\
								&+  \int_0^\cdot   \int_{\R^d} \left(\xi_t(x) - D \xi_t(0) h(x)\right) \nu^X (\d t, \d x),
\end{split}		
\]
yielding the statement.
\end{proof}

\setstretch{1.2}
\section{Change of measure}\label{C}
This appendix collects results on changes of measures.

\begin{theorem}[Girsanov's theorem for absolutely continuous probability measures]\label{T:Girsanov}
 Let $N$ be a $\P$--semimartingale such that 
$$M=\Exp(N)$$ 
is a uniformly integrable $\P$--martingale with $M\ge 0$. Define the probability measure $\Qu$~by 
$$\frac{\d\Qu}{\d \P} =M_\infty.$$ 
For a $\Qu$--semimartingale $V$ and a $\P$--semimartingale $V_{\uparrow}$, $\Qu$--indistinguishable from $V$, the following are equivalent.
\begin{enumerate}[label={\rm(\arabic{*})}, ref={\rm(\arabic{*})}]
\item\label{P:180706:i} $V$ is $\Qu$--special.
\item\label{P:180706:ii} $V_\uparrow +[V_\uparrow,N]$ is $\P$--special.
\end{enumerate} 
If either condition holds then the corresponding compensators are equal, i.e., 
\begin{align*}
	B^V_\Qu = B^{V_\uparrow+[V_\uparrow,N]}, \qquad \text{$\Qu$--almost surely}.
\end{align*}
\end{theorem}
\noindent The proof is quite classical and we do not reproduce it here. For details see \citet[Proposition~5.2]{crIII}.

\begin{theorem}[Girsanov's theorem -- representations]\label{T:Girsanov2}\leavevmode
Assume $\eta,\xi \in \I_{0\R}$ (resp., $\I_{0\Cx}$) with $\eta \geq -1$ are compatible with a semimartingale  $X$ such that 
$\eta\circ X$ is special and $\Delta B^{\eta\circ X} > -1$. Assume further that 
$M=\sfrac{\Exp(\eta\circ X)}{\Exp\left(B^{\eta\circ X}\right)}$ is a real-valued uniformly integrable $\P$--martingale and define the probability measure $\Qu$~by 
$$\frac{\d\Qu}{\d \P} =M_\infty.$$  
Assume also that $\xi \circ X$ is $\Qu$--special. Then the following statements hold.
\begin{enumerate}[label={\rm(\roman{*})}, ref={\rm(\roman{*})}]
\item\label{T:Girsanov2:i} If $X$ is a discrete-time process under $\P$ then $\xi\circ X$ is a discrete-time process under $\Qu$ and
\begin{align*}
B_{\Qu,t}^{\xi\circ X}&
= \sum_{k=1}^{\lfloor t\rfloor}\E_{k-}\Bigg[\xi_{k}(\Delta X_k)\frac{1+\eta_{k}(\Delta X_k)}{1+\E_{k-}[\eta_{k}(\Delta X_k)]}\Bigg],\qquad &t\geq 0.
\end{align*}
\item\label{T:Girsanov2:ii} If $X$ is an It\^o $\P$--semimartingale then $\xi\circ X$ is an It\^o $\Qu$--semimartingale and
\begin{align*}
b^{\xi\circ X}_\Qu = b^{(1+\eta)\xi\circ X}&= D\xi(0) b^{X[h]} 
+ \frac{1}{2} \sum_{i,j = 1}^d \left(D^2_{ij} \xi(0)+2D_i \xi(0)D_j\eta(0)\right)c^X_{ij} \\
&\qquad\qquad\qquad + \int_{\R^d}\left(\xi(x)(1+\eta(x)) - D \xi(0) h(x)\right)F^X(\d x).\\
\end{align*}
\end{enumerate}
\end{theorem}
\vspace{-1.0cm}
\begin{proof}
In this proof all predictable functions appearing in representations are in $\I_{0\R}$ (resp., $\I_{0\Cx}$). By a standard calculation, the process 
$M=\sfrac{\Exp(\eta\circ X)}{\Exp\left(B^{\eta\circ X}\right)}$ satisfies $M = \Exp(N)$ with
$$N =\left(\frac{1+\eta(\id_1)}{1+\id_2}-1\right)\circ (X,B^{\eta\circ X}).$$   
From the representation of quadratic covariation in Example~\ref{E:180809} we then obtain
\begin{equation*}
V+[V,N] = \xi(\id_1)\frac{1+\eta(\id_1)}{1+\id_2}\circ (X,B^{\eta\circ X}) = \xi\frac{1+\eta}{1+\Delta B^{\eta\circ X}}\circ X.
\end{equation*}
The rest follows from the  general Girsanov theorem (Theorem~\ref{T:Girsanov}) and the drift formulae in Theorem~\ref{T:compensators}.
\end{proof}

\begin{corollary}\label{C:compensatorsPIIQ}
With the notation and assumptions as in Theorem~\ref{T:Girsanov2} above, if $X$ is PII under $\P$ stopped at a finite time and $\eta$ is deterministic then $M=\sfrac{\Exp(\eta\circ X)}{\Exp\left(B^{\eta\circ X}\right)}$ is a uniformly integrable martingale. Furthermore, if $\xi$, too, is deterministic then $\xi\circ X$ is PII under $\Qu$ and the following statements hold for all $t\geq 0$.
\begin{enumerate}[label={\rm(\roman{*})}, ref={\rm(\roman{*})}]
\item\label{C:compensatorsPIIQ:i} If $X$ is a discrete-time process under $\P$ then
\begin{align*}
\E^\Qu[\Exp(\xi\circ X)_t]&= \prod_{k=1}^{\lfloor t\rfloor}\E\Bigg[(1+\xi_{k}(\Delta X_k))\frac{1+\eta_{k}(\Delta X_k)}{1+\E[\eta_{k}(\Delta X_k)]}\Bigg].
\end{align*}
\item\label{C:compensatorsPIIQ:ii} If $X$ is an It\^o $\P$--semimartingale then
\begin{align*}
\E^\Qu[\Exp(\xi\circ X)_t] = \exp\left(\int_0^t b_u^{(1+\eta)\xi\circ X}\d u\right).
\end{align*}
\end{enumerate}
\end{corollary}
\begin{proof}
	First note that if $\eta$ is deterministic and if $X$ is PII, then Example~\ref{ex: characteristics} yields that $\eta \circ X$ is again PII. Next, Theorem~\ref{T:PII} yields the martingale property of $M$. The PII property of $X$ under $\Qu$ follows from Girsanov's theorem. The argument then follows from Theorems~\ref{T:PII}, \ref{T:compensators}, and \ref{T:Girsanov2}.
\end{proof}

\end{document}